\documentclass{llncs}
\usepackage{amsmath}
\usepackage{amssymb}
\usepackage{hyperref}

\usepackage{breakurl}

\usepackage{enumerate}
\usepackage{tikz}
\usepackage{boxedminipage}
\usepackage{microtype}
\usetikzlibrary{arrows,shapes,calc}
\hypersetup{hypertexnames=false}

\def\Z{{\mathbb Z}}

\newcommand{\problemdef}[3]{
	\begin{center}
		\begin{boxedminipage}{.99\textwidth}
			\textsc{{#1}}\\[2pt]
			\begin{tabular}{ r p{0.8\textwidth}}
				\textit{~~~~Instance:} & {#2}\\
				\textit{Question:} & {#3}
			\end{tabular}
		\end{boxedminipage}
	\end{center}
}

\newcommand{\NP}{{\sf NP}}
\newcommand{\FPT}{{\sf FPT}}

\title{Independent Feedback Vertex Set for $P_5$-free Graphs\thanks{This paper received support from EPSRC (EP/K025090/1), London Mathematical Society (41536), the Leverhulme Trust (RPG-2016-258) and Fondation Sciences Math\'ematiques de Paris.}}

\author{Marthe Bonamy\inst{1} \and
Konrad K. Dabrowski\inst{2} \and
Carl Feghali\inst{3}
\and\\ Matthew Johnson\inst{2}
\and Dani\"el Paulusma\inst{2}}
\institute{
CNRS, LaBRI, France \texttt{marthe.bonamy@u-bordeaux.fr},
\and
School of Engineering and Computing Sciences, Durham University, UK\\
\texttt{\{konrad.dabrowski,matthew.johnson2,daniel.paulusma\}@durham.ac.uk},
\and
IRIF \& Universit\'e Paris Diderot, Paris France \texttt{feghali@irif.fr}}

\newcommand\displaycase[1]{{\bf #1}}
\newcommand{\qedllncs}{\qed}
\newcounter{ctrclaim}[theorem]
\newcounter{ctrcase}[theorem]
\newcommand{\clm}[1]{\setcounter{ctrcase}{0}\medskip\phantomsection\refstepcounter{ctrclaim}\noindent\displaycase{Claim \thectrclaim. }{\em #1}\\}
\newcommand{\clmnonewline}[1]{\setcounter{ctrcase}{0}\medskip\phantomsection\refstepcounter{ctrclaim}\noindent\displaycase{Claim \thectrclaim. }{\em #1}}
\newcommand{\thmcase}[1]{\medskip\phantomsection\refstepcounter{ctrcase}\noindent\displaycase{Case \thectrcase. }{\em #1}\\}

\begin{document}
\maketitle

\begin{abstract}
The \NP-complete problem {\sc Feedback Vertex Set} is that of deciding whether or not it is possible, for a given integer $k\geq 0$, to delete at most~$k$ vertices from a given graph so that what remains is a forest.
The variant in which the deleted vertices must form an independent set is called {\sc Independent Feedback Vertex Set} and is also \NP-complete.
In fact, even deciding if an independent feedback vertex set exists is \NP-complete and this problem is closely related to the {\sc $3$-Colouring} problem, or equivalently, to the problem of deciding whether or not a graph has an independent odd cycle transversal, that is, an independent set of vertices whose deletion makes the graph bipartite.
We initiate a systematic study of the complexity of {\sc Independent Feedback Vertex Set} for $H$-free graphs.
We prove that it is \NP-complete if~$H$ contains a claw or cycle.
Tamura, Ito and Zhou proved that it is polynomial-time solvable for $P_4$-free graphs.
We show that it remains polynomial-time solvable for $P_5$-free graphs.
We prove analogous results for the {\sc Independent Odd Cycle Transversal} problem, which asks whether or not a graph has an independent odd cycle transversal of size at most~$k$ for a given integer $k\geq 0$.
Finally, in line with our underlying research aim, we compare the complexity of {\sc Independent Feedback Vertex Set} for $H$-free graphs with the complexity of {\sc $3$-Colouring}, {\sc Independent Odd Cycle Transversal} and other related problems.
\end{abstract}

\section{Introduction}\label{s-intro}

Many computational problems in the theory and application of graphs can be formulated as modification problems: from a graph~$G$, some other graph~$H$ with a desired property must be obtained using certain permitted operations.
The number of graph operations used (or some other measure of cost) must be minimised.
The computational complexity of a graph modification problem depends on the desired property, the operations allowed and the possible inputs; that is, we can prescribe the class of graphs to which~$G$ must belong.
This leads to a rich variety of different problems, which makes graph modification a central area of research in algorithmic graph theory.

A set~$S$ of vertices in a graph~$G$ is a {\em feedback vertex set} of~$G$ if removing the vertices of~$S$ results in an acyclic graph, that is, the graph $G-S$ is a forest.
The {\sc Feedback Vertex Set} problem asks whether or not a graph has a feedback vertex set of size at most~$k$ for some integer $k\geq 0$ and is a well-known example of a graph modification problem: the desired property is that the obtained graph is acyclic and the permitted operation is vertex deletion.
The directed variant of {\sc Feedback Vertex Set} was one of the original problems proven to be \NP-complete by Karp~\cite{Ka72}.
The proof of this implies \NP-completeness of the undirected version even for graphs of maximum degree~$4$ (see~\cite{GJ79}).
We refer to the survey of Festa et~al.~\cite{FPR99} for further details of this classic problem.

In this paper, we consider the problem where we require the feedback vertex set to be an independent set.
We call such a set an {\em independent feedback vertex set}.
We have the following decision problem.

\problemdef{\sc Independent Feedback Vertex Set}{a graph~$G$ and an integer $k\geq 0$.}{does~$G$ have an independent feedback vertex set of size at most~$k$?}
Many other graph problems have variants with an additional constraint that a set of vertices must be independent.
For example, see~\cite{GH13} for a survey on {\sc Independent Dominating Set}, and~\cite{MOR13} for {\sc Independent Odd Cycle Transversal}, also known as {\sc Stable Bipartization}.
An {\em independent odd cycle transversal} of a graph~$G$ is an independent set~$S$ such that $G-S$ is bipartite, and the latter problem is that of deciding whether or not a graph has such a set of size at most~$k$ for some given integer~$k$.

We survey known results on {\sc Independent Feedback Vertex Set} below.

\subsection{Related Work}

Not every graph admits an independent feedback vertex set (consider complete graphs on at least four vertices).
Graphs that do admit an independent feedback vertex set are said to be {\em near-bipartite}, and we can ask about the decision problem of recognising such graphs.
\problemdef{\sc Near-Bipartiteness}{a graph~$G$.}{is~$G$ near-bipartite (that is, does~$G$ have an independent feedback vertex set)?}
{\sc Near-Bipartiteness} is \NP-complete even for graphs of maximum degree~$4$~\cite{YY06} and for graphs of diameter~$3$~\cite{BDFJP17} (see~\cite{BDFJP17c} for a proof).
Hence, by setting $k=n$, we find that {\sc Independent Feedback Vertex Set} is \NP-complete for these two graph classes.
The {\sc Independent Feedback Vertex Set} problem is even \NP-complete for planar bipartite graphs of maximum degree~$4$ (see~\cite{TIZ15}).
As bipartite graphs are near-bipartite, this result shows that there are classes of graphs where {\sc Independent Feedback Vertex Set} is harder than {\sc Near-Bipartiteness}.
To obtain tractability results for {\sc Independent Feedback Vertex Set}, we need to make some further assumptions.

One way is to consider the problem from a parameterized point of view.
Taking~$k$ as the parameter, Misra et~al.~\cite{MPRS12} proved that {\sc Independent Feedback Vertex Set} is fixed-parameter tractable by giving a cubic kernel.
This is in line with the fixed-parameter tractability of the general {\sc Feedback Vertex Set} problem (see~\cite{KP14} for the fastest known \FPT\ algorithm).
Later, Agrawal et~al.~\cite{AGSS16} gave a faster \FPT\ algorithm for {\sc Independent Feedback Vertex Set} and also obtained an upper bound on the number of minimal independent feedback vertex sets of a graph.

Another way to obtain tractability results is to restrict the input to special graph classes in order to determine graph properties that make the problem polynomial-time solvable.
We already mentioned some classes for which {\sc Independent Feedback Vertex Set} is \NP-complete.
In a companion paper~\cite{BDFJP17c}, we show that the problem is polynomial-time solvable for graphs of diameter~$2$, and as stated above, the problem is \NP-complete on graphs of diameter~$3$.
Tamura et~al.~\cite{TIZ15} showed that {\sc Independent Feedback Vertex Set} is polynomial-time solvable for chordal graphs, graphs of bounded treewidth and for cographs.
The latter graphs are also known as $P_4$-free graphs ($P_r$ denotes the path on~$r$ vertices and a graph is {\em $H$-free} if it has no induced subgraph isomorphic to~$H$), and this strengthened a result of Brandst\"adt et~al.~\cite{BBKNP13}, who proved that {\sc Near-Bipartiteness} is polynomial-time solvable for $P_4$-free graphs.

\subsection{Our Contribution}
The {\sc Independent Feedback Vertex Set} problem is equivalent to asking for a (proper) $3$-colouring of a graph, such that one colour class has at most~$k$ vertices and the union of the other two induces a forest.
We wish to compare the behaviour of {\sc Independent Feedback Vertex Set} with that of the {\sc $3$-Colouring} problem.
It is well known that the latter problem is also \NP-complete~\cite{Lo73} in general and polynomial-time solvable on many graph classes (see, for instance, the surveys~\cite{GJPS,RS04b}).
We also observe that {\sc $3$-Colouring} is equivalent to asking whether or not a graph has an independent odd cycle transversal (of any size).
However, so far very few graph classes are known for which {\sc Independent Feedback Vertex Set} is tractable and our goal is to find more of them.
For this purpose, we consider $H$-free graphs and extend the result~\cite{TIZ15} for $P_4$-free graphs in a systematic way.

In Section~\ref{s-hard}, we consider the cases where~$H$ contains a cycle or a claw.
We first prove that {\sc Near-Bipartiteness}, and thus {\sc Independent Feedback Vertex Set}, is \NP-complete on line graphs, which form a subclass of the class of claw-free graphs.
We then prove that {\sc Independent Feedback Vertex Set} is \NP-complete for graphs of arbitrarily large girth.
Together, these results imply that {\sc Independent Feedback Vertex Set} is \NP-complete for $H$-free graphs if~$H$ contains a cycle or claw.
Hence, only the cases where~$H$ is a {\em linear forest}, that is, a disjoint union of paths, remain open.
In particular, the case where~$H$ is a single path has not yet been resolved.
Due to the result of~\cite{TIZ15} for $P_4$-free graphs, the first open case to consider is when $H=P_5$ (see also \figurename~\ref{fig:P4P5}).

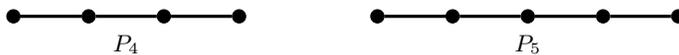
\begin{figure}
\begin{center}
\tikzstyle{vertex}=[circle,draw=black, fill=black, minimum size=5pt, inner sep=1pt]
\tikzstyle{edge} =[draw,black]
\begin{tabular}{cc}
\begin{minipage}{0.4\textwidth}
\begin{center}
\begin{tikzpicture}
\foreach \num in {0,1,2,3} \node[vertex] (x0\num) at (\num,0) {};
\foreach \source/\dest in {x00/x01, x01/x02, x02/x03} \path[edge, black, very thick] (\source) -- (\dest);
\end{tikzpicture}
\end{center}
\end{minipage}
&
\begin{minipage}{0.4\textwidth}
\begin{center}
\begin{tikzpicture}
\foreach \num in {0,1,2,3,4} \node[vertex] (x0\num) at (\num,0) {};
\foreach \source/\dest in {x00/x01, x01/x02, x02/x03, x03/x04} \path[edge, black, very thick] (\source) -- (\dest);
\end{tikzpicture}
\end{center}
\end{minipage}\\
$P_4$ & $P_5$
\end{tabular}
\end{center}
\caption{\label{fig:P4P5}The paths on four and five vertices.}
\end{figure}

The class of $P_5$-free graphs is a well-studied graph class.
For instance, Ho{\`{a}}ng et~al.~\cite{HKLSS10} proved that for every integer~$k$, {\sc $k$-Colouring} is polynomial-time solvable for $P_5$-free graphs, whereas Golovach and Heggernes~\cite{GH09} showed that {\sc Choosability} is fixed-parameter tractable for $P_5$-free graphs when parameterized by the size of the lists of admissible colours.
Lokshantov et~al.~\cite{LVV14} solved a long-standing open problem by giving a polynomial-time algorithm for {\sc Independent Set} restricted to $P_5$-free graphs (recently, their result was extended to $P_6$-free graphs by Grzesik et~al.~\cite{GKPP17}).

Our main result is that {\sc Independent Feedback Vertex Set} is polynomial-time solvable for $P_5$-free graphs.
This is proved in Sections~\ref{s-list} and~\ref{s-ind}: in Section~\ref{s-list} we give a polynomial-time algorithm for {\sc Near-Bipartiteness} on $P_5$-free graphs, and in Section~\ref{s-ind} we show how to {\em extend} this algorithm to solve {\sc Independent Feedback Vertex Set} in polynomial time for $P_5$-free graphs.

\begin{sloppypar}
In Section~\ref{a-d} we consider the related problem {\sc Independent Odd Cycle Transversal}.
We prove that all our results for {\sc Independent Feedback Vertex Set} also hold for {\sc Independent Odd Cycle Transversal}.
\end{sloppypar}

In Section~\ref{s-con}, we compare the complexities of {\sc Independent Feedback Vertex Set} and {\sc Independent Odd Cycle Transversal} for $H$-free graphs with the complexity of {\sc $3$-Colouring} and several other related problems, such as {\sc Feedback Vertex Set}, {\sc Vertex Cover}, {\sc Independent Vertex Cover} and {\sc Dominating Induced Matching}.
We also survey some related open problems.

\section{Hardness When~$H$ Contains a Cycle or Claw}\label{s-hard}

Before stating the results in this section, we first introduce some necessary terminology.
The {\em line graph}~$L(G)$ of a graph $G=(V,E)$ has the edge set~$E$ of~$G$ as its vertex set, and two vertices~$e_1$ and~$e_2$ of~$L(G)$ are adjacent if and only if~$e_1$ and~$e_2$ share a common end-vertex in~$G$.
The {\em claw} is the graph shown in \figurename~\ref{fig:claw}.
It is well known and easy to see that every line graph is claw-free.
A graph is {\em (sub)cubic} if every vertex has (at most) degree~$3$.

We first prove that {\sc Near-Bipartiteness} is \NP-complete for line graphs.
It was already known that {\sc Feedback Vertex Set} is \NP-complete for line graphs of planar cubic bipartite graphs~\cite{Mu17}.

\begin{figure}
\begin{center}
\tikzstyle{vertex}=[circle,draw=black, fill=black, minimum size=5pt, inner sep=1pt]
\tikzstyle{edge} =[draw,black]
\begin{tikzpicture}
\node[vertex] (x) at (0:0) {};
\node[vertex] (a) at (90:1) {};
\node[vertex] (b) at (90+120:1) {};
\node[vertex] (c) at (90+240:1) {};
\foreach \source in {a,b,c} \path[edge, black, very thick] (\source) -- (x);
\end{tikzpicture}
\end{center}
\caption{\label{fig:claw}The claw.}
\end{figure}
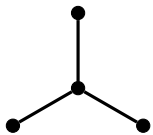

\begin{theorem}\label{t-hard}
{\sc Near-Bipartiteness} is \NP-complete for line graphs of planar subcubic bipartite graphs.
\end{theorem}

\begin{proof}
\setcounter{ctrclaim}{0}
As the problem is readily seen to be in \NP, it suffices to prove \NP-hardness.
We reduce from the {\sc Hamilton Cycle Through Specified Edge} problem.
Given a graph~$G$ and an edge~$e$ of~$G$, this problem asks whether~$G$ has a Hamilton cycle through~$e$.
Labarre~\cite{La11} observed that {\sc Hamilton Cycle Through Specified Edge} is \NP-complete for planar cubic bipartite graphs
by noting that it follows easily from the analogous result of Akiyama et~al.~\cite{ANS80} for {\sc Hamilton Cycle}.
Let~$(G,e)$ be an instance of {\sc Hamilton Cycle Through Specified Edge}.
By the aforementioned result, we may assume that~$G$ is a planar cubic bipartite graph.
Let~$u_1$ and~$u_2$ be the two end-vertices of~$e$.
Delete the edge~$u_1u_2$ and add two new vertices~$v_1$ and~$v_2$, and two new edges~$e_1=u_1v_1$ and~$e_2=v_2u_2$.
Let~$G$ be the resulting graph.
We note that both~$v_1$ and~$v_2$ have degree~$1$ in~$G'$, whereas every other vertex of~$G'$ has degree~$3$.
Hence~$G'$ is subcubic.
Since~$G$ is planar and bipartite, it follows that~$G'$ is planar and bipartite.
Moreover, we make the following observation.

\clmnonewline{\label{clm:p-hard2:1}The graph~$G$ has a Hamilton cycle through~$e$ if and only if~$G'$ has a Hamilton path from~$v_1$ to~$v_2$.}

\medskip
\noindent
Let~$n$ be the number of vertices in~$G$.
Then the number of vertices in~$G'$ is $n+2$, meaning that a Hamilton path in~$G'$ has $n+1$ edges.
Moreover, as~$G$ is cubic, $G$ has~$\frac{3}{2}n$ edges, so~$G'$ has $\frac{3}{2}n+1$ edges, implying that~$L(G')$ has $\frac{3}{2}n+1$ vertices.
Furthermore, $e_1$ and~$e_2$ have degree~$2$ in~$L(G')$ and all other vertices in~$L(G')$ have degree~$4$, so~$L(G')$ has maximum degree at most~$4$.
 
\clm{\label{clm:p-hard2:2}The graph~$G'$ has a Hamilton path from~$v_1$ to~$v_2$ if and only if~$L(G')$ is near-bipartite.}
To prove Claim~\ref{clm:p-hard2:2}, first suppose that~$G'$ has a Hamilton path~$P$ from~$v_1$ to~$v_2$.
Then, as every vertex in~$G'$ apart from~$v_1$ and~$v_2$ has degree~$3$, it follows that $G'-E(P)$ consists of two isolated vertices ($v_1$ and~$v_2$) and a set~$S$ of isolated edges.
Thus~$S$ is an independent set in~$L(G')$, and so~$L(G')$ is near-bipartite.

Now suppose that~$L(G')$ is near-bipartite.
Then~$L(G')$ contains a set~$S$ of vertices, such that $F=L(G')-S$ is a forest.
As~$L(G')$ is a line graph,~$L(G')$ is claw-free.
This means that~$F$ is the disjoint union of one or more paths.
Suppose that~$F$ contains more than one path.
Then, as~$e_1$ and~$e_2$ are the only two vertices in~$L(G')$ that are of degree~$2$ and all other vertices of~$L(G')$ have degree~$4$, at least one path of~$F$ has an end-vertex of degree~$4$ in~$L(G')$.
Let~$f$ be this vertex.
As~$f$ is the end-vertex of a path in~$F$, we find that~$f$ has three neighbours~$f_1$, $f_2$,~$f_3$ in~$S$.
As~$S$ is an independent set, $\{f,f_1,f_2,f_3\}$ induces a claw in~$L(G')$.
This contradiction tells us that~$F$ consists of exactly one path~$P$, and by the same reasoning, $e_1$ and~$e_2$ must be the end-vertices of~$P$.

Hence every vertex of~$P$ apart from its end-vertices has two neighbours in~$S$ and every vertex in~$S$ has four neighbours on~$P$.
Moreover, $e_1$ and~$e_2$ have exactly one neighbour in~$S$.
This means that $1+1+2(|V(P)|-2)=4|S|$, so $|S|=\frac{1}{2}(|V(P)|-1)$.
Hence we find that
\[
\begin{array}{lcl}
\frac{3}{2}(|V(P)|-1)+1 &= &|V(P)|+\frac{1}{2}(|V(P)|-1)\\[5pt]
&= &|V(P)|+|S|\\[5pt]
&= &|L(G')|\\[5pt]
&= &\frac{3}{2}n+1,
\end{array}
\]
so $|V(P)|=n+1$.
Hence, as~$G'$ has $n+2$ vertices, the pre-image of~$P$ in~$G'$ is a Hamilton path of~$G'$ with end-vertices~$v_1$ and~$v_2$.

\medskip
\noindent
By combining Claims~\ref{clm:p-hard2:1} and~\ref{clm:p-hard2:2} we have completed our hardness reduction and the theorem is proved.\qedllncs
\end{proof}

Theorem~\ref{t-hard} has the following immediate consequence (take $k=n$).

\begin{corollary}\label{c-hard}
{\sc Independent Feedback Vertex Set} is \NP-complete for line graphs of planar subcubic bipartite graphs.
\end{corollary}

We will now prove that {\sc Independent Feedback Vertex Set} is \NP-complete for graphs with no small cycles even if their maximum degree is small.
The {\em length} of a cycle~$C$ is the number of edges of~$C$.
The {\em girth}~$g(G)$ of a graph~$G$ is the length of a shortest cycle of~$G$; if~$G$ has no cycles then $g(G)=\infty$.
The {\em subdivision} of an edge $e=uv$ in a graph deletes~$e$ and adds a new vertex~$w$ and edges~$uw$ and~$wv$.
We first need the following observation, which is well known.
For completeness we give a short proof.

\begin{lemma}[see e.g.~\cite{MPRS12}]\label{l-wellknown}
Let~$uv$ be an edge in a graph~$G$.
Let~$G'$ be the graph obtained from~$G$ after subdividing~$uv$.
Then~$G$ has a feedback vertex set of size at most~$k$ if and only if~$G'$ does.
\end{lemma}

\begin{proof}
Let~$w$ denote the new vertex obtained from subdividing~$uv$.
Any feedback vertex set~$S$ of~$G$ is a feedback vertex set of~$G'$.
Suppose~$S'$ is a feedback vertex set of~$G'$.
If $w\notin S'$, then~$S'$ is a feedback vertex set of~$G$.
Suppose $w\in S'$.
If at least one of~$u$ and~$v$ are in~$S'$ as well, then $S'\setminus\{w\}$ is a feedback vertex set of~$G$.
If neither~$u$ nor~$v$ belong to~$S'$, then $(S'\setminus \{w\})\cup \{u\})$ is a feedback vertex set of~$G$ with the same size as~$S'$.\qedllncs
\end{proof}

Lemma~\ref{l-wellknown} implies that {\sc Feedback Vertex Set} is \NP-complete for graphs of girth at least~$g$ for every constant $g\geq 3$.
We also use this lemma to prove our next result.

\begin{proposition}\label{p-hard1}
For every constant~$g\geq 3$, {\sc Independent Feedback Vertex Set} is \NP-complete for graphs of maximum degree at most~$4$ and girth at least~$g$.
\end{proposition}

\begin{proof}
For a graph~$G$, let~$G_s$ be the graph obtained from~$G$ after subdividing every edge of~$G$; we say that~$G_s$ is a {\em subdivided} copy of~$G$.
Let~${\cal G}_s$ be the graph class obtained from a graph class~${\cal G}$ after replacing each $G\in {\cal G}$ by its subdivided copy~$G_s$.
It follows from Lemma~\ref{l-wellknown} that if {\sc Feedback Vertex Set} is \NP-complete for some graph class~${\cal G}$, then it is also \NP-complete for~${\cal G}_s$.
By starting from the fact that {\sc Feedback Vertex Set} is \NP-complete for line graphs of planar cubic bipartite graphs~\cite{Mu17} and applying this observation a sufficient number of times, we find that for any constant $g\geq 3$, {\sc Feedback Vertex Set} is \NP-complete for graphs of maximum degree at most~$4$ and girth at least~$g$.
Moreover, any non-independent feedback vertex set~$S$ of a subdivided copy~$G_s$ of a graph~$G$ contains two adjacent vertices, one of which has degree~$2$ in~$G_s$.
Hence, we can remove such a degree~$2$ vertex from~$S$ to obtain a smaller feedback vertex set of~$G_s$.
Thus all minimum feedback vertex sets of~$G_s$ are independent.
This observation, which can also be found in~\cite{MPRS12}, together with \NP-hardness for {\sc Feedback Vertex Set} for graphs of maximum degree at most~$4$ and with arbitrarily large girth proves the proposition.\qedllncs
\end{proof}

Recall that every line graph is claw-free.
We also observe that for a graph~$H$ with a cycle~$C$, the class of graphs of girth at least $|C|+1$ is a subclass of the class of $H$-free graphs.
Hence, we can combine Corollary~\ref{c-hard} and Proposition~\ref{p-hard1} to obtain the following result.

\begin{corollary}\label{c-hard2}
Let~$H$ be a graph that contains a claw or a cycle.
Then {\sc Independent Feedback Vertex Set} is \NP-complete for $H$-free graphs of maximum degree at most~$4$.
\end{corollary}

\section{Near-Bipartiteness of $P_5$-free Graphs}\label{s-list}

In this section, we show that {\sc Near-Bipartiteness} is polynomial-time solvable for $P_5$-free graphs, i.e. we give a polynomial-time algorithm for testing whether or not a $P_5$-free graph has an independent feedback vertex set.
To obtain a {\em minimum} feedback vertex set we need to first run this algorithm and then do the additional work described in Section~\ref{s-ind}.

Our algorithm in this section solves a slightly more general problem, which is a special variant of {\sc List $3$-Colouring}.
In the {\sc List $3$-Colouring} problem each vertex~$v$ is assigned a subset~$L(v)$ of colours from $\{1,2,3\}$ and we must verify whether or not a $3$-colouring exists in which each vertex~$v$ is coloured with a colour from~$L(v)$.
We say that a $3$-colouring of a graph~$G$ is {\em semi-acyclic} if the vertices coloured~$2$ or~$3$ induce a forest, and we note that~$G$ has such a colouring if and only if~$G$ is near-bipartite.
This leads to the following variant of {\sc List $3$-Colouring}.

\problemdef{\sc List Semi-Acyclic $3$-Colouring}{a graph~$G$ and a function $L:V(G) \rightarrow \{S \; | \; S \subseteq \{1,2,3\}\}$.}{does~$G$ have a semi-acyclic $3$-colouring~$c$ such that $c(v) \in L(v)$ for all $v \in V(G)$?}

A graph~$G$ is near-bipartite if and only if $(G,L)$, with $L(v)=\{1,2,3\}$ for all $v \in V(G)$, is a \texttt{yes}-instance of {\sc List Semi-Acyclic $3$-Colouring}.
To recognise near-bipartite $P_5$-free graphs in polynomial time, we will show the stronger statement that {\sc List Semi-Acyclic $3$-Colouring} is polynomial-time solvable for $P_5$-free graphs.

A set of vertices in a graph~$G$ is {\em dominating} if every vertex of~$G$ is either in the set or has at least one neighbour in it.
We will use a lemma of Bacs\'o and Tuza.

\begin{lemma}[\cite{BT90}]\label{lem:P5free-clique-P3-dom}
Every connected $P_5$-free graph admits a dominating set that induces either a clique or a~$P_3$.
\end{lemma}

\noindent
Lemma~\ref{lem:P5free-clique-P3-dom} implies that every connected $3$-colourable $P_5$-free graph has a dominating set of size at most~$3$ (since it has no clique on more than three vertices).
This was used by Randerath et~al.~\cite{RST02} to show that {\sc $3$-Colouring} is polynomial-time solvable on $P_5$-free graphs.
Their algorithm tries all possible $3$-colourings of a dominating set of size at most~$3$.
It then adjusts the lists of the other vertices (which were originally set to $\{1,2,3\}$) to lists of size at most~$2$.
As shown by Edwards~\cite{Ed86}, {\sc $2$-{List Colouring}} can be translated to an instance of {\sc $2$-Satisfiability}, which is well known and readily seen to be solvable in linear time.
Hence this approach results in a polynomial (even constant) number of instances of the {\sc $2$-Satisfiability} problem.

Our goal is also to apply Lemma~\ref{lem:P5free-clique-P3-dom} on a connected $P_5$-free graph~$G$ and to reduce an instance $(G,L)$ of {\sc List Semi-Acyclic $3$-Colouring} to a polynomial number of instances of {\sc $2$-Satisfiability}.
However, in our case this is less straightforward than in the case of {\sc $3$-Colouring} restricted to $P_5$-free graphs: the restriction of {\sc List Semi-Acyclic $3$-Colouring} to lists of size~$2$ turns out to be \NP-complete for general graphs even if every list consists of either colours~$1$ and~$3$ or only colour~$2$.

\begin{sloppypar}
\begin{theorem}\label{t-2sat}
{\sc List Semi-Acyclic $3$-Colouring} is \NP-complete even if $L(v)\in \{\{1,3\},\{2\}\}$ for every vertex~$v$ in the input graph.
\end{theorem}
\end{sloppypar}

\begin{proof}
The problem is clearly in \NP\ so we need only show that it is \NP-hard.
We do this by reduction from {\sc Satisfiability}.

Let~$\phi$ be an instance of {\sc Satisfiability}.
Note that we can assume that, for each variable~$x$, $\phi$ contains both the literals~$x$ and~$\overline{x}$, and that each clause contains more than one literal (otherwise in polynomial-time we can obtain another smaller instance whose satisfiability is the same as that of~$\phi$).
We create an instance $(G,L)$ of {\sc List Semi-Acyclic $3$-Colouring} as follows (see also \figurename~\ref{fig:NPC-GL}):
\begin{figure}
\begin{center}
\scalebox{0.90}{
\tikzstyle{vertex}=[circle,draw=black, fill=black, minimum size=5pt, inner sep=1pt]
\tikzstyle{edge} =[draw,black]
\begin{tikzpicture}[scale=1]

   \foreach \pos/\name / \label / \posn / \dist in {
{(30:1.2)/c11/$\{2\}$/{left}/0},
{(150:1.2)/c12/$\{2\}$/{right}/0},
{(270:1.5)/c13/$\{2\}$/{below}/0}}
       { \node[vertex] (\name) at \pos {};
       \node [\posn=\dist, align=center] at (\name) {\label};
       }

   \foreach \pos/\name / \label / \posn / \dist / \list / \lposn / \ldist in {
{(90:1.5)/x11/$\overline{x_1}$/{below}/1.7/$\{1,3\}$/{above}/0},
{(210:1.2)/x21/$\{1,3\}$/{right}/0/$v_{x_2}$/{below left}/-1.5},
{(330:1.2)/x31/$v_{x_3}$/{above right}/0/$\{1,3\}$/{left}/0}}
       { \node[vertex] (\name) at \pos {};
       \node [\posn=\dist, align=center] at (\name) {\label};
       \node [\lposn=\ldist, align=center] at (\name) {\list};
       }

\foreach \source/ \dest  in {x11/c11, x11/c12, x21/c13, x21/c12, x31/c11, x31/c13}
       \path[edge, black,  very thick] (\source) --  (\dest);

\begin{scope}[xshift=3.5cm]

   \foreach \pos/\name / \label / \posn / \dist in {
{(30:1.2)/c21/$\{2\}$/{left}/0},
{(150:1.2)/c22/$\{2\}$/{right}/0},
{(270:1.5)/c23/$\{2\}$/{below}/0}}
       { \node[vertex] (\name) at \pos {};
       \node [\posn=\dist, align=center] at (\name) {\label};
       }

   \foreach \pos/\name / \label / \posn / \dist / \list / \lposn / \ldist in {
{(90:1.5)/x12/$v_{x_1}$/{below}/1.7/$\{1,3\}$/{above}/0},
{(210:1.2)/x32/$\overline{x_3}$/{below left}/0/$\{1,3\}$/{right}/0},
{(330:1.2)/x42/$v_{x_4}$/{below right}/0/$\{1,3\}$/{left}/0}}
       { \node[vertex] (\name) at \pos {};
       \node [\posn=\dist, align=center] at (\name) {\label};
       \node [\lposn=\ldist, align=center] at (\name) {\list};
       }

\end{scope}

\foreach \source/ \dest  in {x12/c21, x12/c22, x32/c23, x32/c22, x42/c21, x42/c23}
       \path[edge, black,  very thick] (\source) --  (\dest);

 \path[edge, red, very thick] (x11) --  (x12);
 \path[edge, red, very thick] (x31) --  (x32);

\begin{scope}[xshift=7.0cm]
   \foreach \pos/\name / \label / \posn / \dist in {
{(30:1.2)/c31/$\{2\}$/{left}/0},
{(150:1.2)/c32/$\{2\}$/{right}/0},
{(270:1.5)/c33/$\{2\}$/{below}/0}}
       { \node[vertex] (\name) at \pos {};
       \node [\posn=\dist, align=center] at (\name) {\label};
       }

   \foreach \pos/\name / \label / \posn / \dist / \list / \lposn / \ldist in {
{(90:1.5)/x13/${x_1}$/{below}/1.7/$\{1,3\}$/{above}/0},
{(210:1.2)/x43/$\overline{x_4}$/{below left}/0/$\{1,3\}$/{right}/0},
{(330:1.2)/x53/$\overline{x_5}$/{below right}/-0.5/$\{1,3\}$/{left}/0}}
       { \node[vertex] (\name) at \pos {};
       \node [\posn=\dist, align=center] at (\name) {\label};
       \node [\lposn=\ldist, align=center] at (\name) {\list};
       }
\end{scope}

\node[vertex] (mid1) at ($ (x12) !.5! (x13) $) {};
       \node [above, align=center] at (mid1) {$\{1,3\}$};

\foreach \source/ \dest  in {x13/c31, x13/c32, x43/c33, x43/c32, x53/c31, x53/c33}
       \path[edge, black,  very thick] (\source) --  (\dest);

  \path[edge, red, very thick] (x42) --  (x43);
  \path[edge, blue, very thick] (x12) --  (mid1);
  \path[edge, blue, very thick] (x13) --  (mid1);

\begin{scope}[xshift=10.5cm]
   \foreach \pos/\name / \label / \posn / \dist in {
{(30:1.2)/c41/$\{2\}$/{left}/0},
{(150:1.2)/c42/$\{2\}$/{right}/0},
{(270:1.5)/c43/$\{2\}$/{left}/0}}
       { \node[vertex] (\name) at \pos {};
       \node [\posn=\dist, align=center] at (\name) {\label};
       }

   \foreach \pos/\name / \label / \posn / \dist / \list / \lposn / \ldist in {
{(90:1.5)/x24/$\overline{x_2}$/{below}/1.7/$\{1,3\}$/{above}/0},
{(210:1.2)/x54/$v_{x_5}$/{below left}/-1.5/$\{1,3\}$/{right}/0},
{(330:1.2)/x34/${x_3}$/{right}/0/$\{1,3\}$/{left}/0}}
       { \node[vertex] (\name) at \pos {};
       \node [\posn=\dist, align=center] at (\name) {\label};
       \node [\lposn=\ldist, align=center] at (\name) {\list};
       }
\end{scope}

\foreach \source/ \dest  in {x24/c41, x24/c42, x54/c43, x54/c42, x34/c41, x34/c43}
       \path[edge, black,  very thick] (\source) --  (\dest);

  \path[edge, red, very thick] (x54) --  (x53);

\path[edge, red, very thick] (x21) to[out=135,in=190] (-0.3,2.5) to[out=0,in=160]	(x24);

\node[vertex] (mid2) at (5.25,-2.4) {};
       \node [below, align=center] at (mid2) {$\{1,3\}$};

\path[edge, blue, very thick] (x31) .. controls (2,-2) and (4,-2.4) .. (mid2);

\path[edge, blue, very thick] (x34) .. controls (13,-2.2) and (6.5,-2.4) .. (mid2);

\end{tikzpicture}
}
\end{center}
\caption{The instance $(G,L)$ of {\sc List Semi-Acyclic $3$-Colouring} formed from the following instance of {\sc Satisfiability}: $(\overline{x_1} \vee x_2 \vee x_3) \wedge (x_1 \vee \overline{x_3} \vee x_4) \wedge (x_1 \vee \overline{x_4} \vee \overline{x_5}) \wedge (\overline{x_2} \vee {x_3} \vee x_5)$.
The list for each vertex is displayed and literals label the vertices that represent them except that, for each variable~$x$, one vertex is labelled~$v_x$.}\label{fig:NPC-GL}
\end{figure}
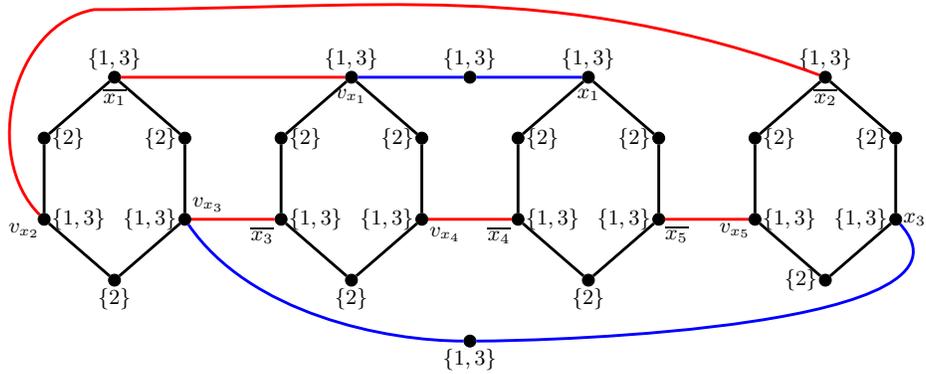
\begin{itemize}
\item For each clause~$C$ of~$\phi$ create a $(2|C|)$-cycle and assign lists $\{1,3\}$ and~$\{2\}$ alternately to vertices around the cycle.
Let the literals of the clause be represented by distinct vertices with lists $\{1,3\}$.
\item For each variable~$x$, choose a clause containing the positive literal~$x$ and let~$v_x$ be the vertex representing~$x$ in the corresponding cycle.
For every other occurrence (if there are any) of the positive literal~$x$, let the corresponding vertex be adjacent to a new \emph{middle} vertex that is also joined by an edge to~$v_x$.
Assign the list $\{1,3\}$ to the middle vertex.
For every occurrence of the negative literal~$\overline{x}$, add an edge so that the corresponding vertex is adjacent to~$v_x$.
\end{itemize}
We claim that $(G,L)$ is a yes-instance of {\sc List Semi-Acyclic $3$-Colouring} if and only if~$\phi$ is a \texttt{yes}-instance of {\sc Satisfiability}.

First suppose that~$G$ has a semi-acyclic $3$-colouring that respects~$L$, and let us show that a satisfying assignment for~$\phi$ can be found.
For each variable~$x$, if~$v_x$ is coloured~$1$, then let~$x$ be true; if it is coloured~$3$, let~$x$ be false.
Note that every other vertex corresponding to the positive literal~$x$ must be coloured the same as~$x$, and every vertex corresponding to an instance of~$\overline{x}$ is coloured differently, so each literal is coloured~$1$ if and only if it is true.
Thus every clause contains a true literal, otherwise in the corresponding cycle every vertex would be coloured~$2$ or~$3$ and the colouring would not be semi-acyclic.

Now suppose that~$\phi$ has a satisfying assignment.
If a literal is true in this assignment, colour the corresponding vertex~$1$, otherwise colour it~$3$.
Colour each middle vertex with the colour not used on its neighbours (which must be coloured alike).
Clearly this is a $3$-colouring, and each cycle corresponding to a clause contains a vertex coloured~$1$ as it contains a true literal.
No other cycle in the graph is coloured with only~$2$ and~$3$, as the only edges that do not belong to the cycles representing the clauses each join a vertex coloured~$1$ to a vertex coloured~$3$.
Thus the colouring is semi-acyclic.\qedllncs
\end{proof}

By Theorem~\ref{t-2sat}, to prove that {\sc List Semi-Acyclic $3$-Colouring} is polynomial-time solvable on $P_5$-free graphs, we need to refine our analysis and exploit $P_5$-freeness beyond the use of Lemma~\ref{lem:P5free-clique-P3-dom}.
We adapt the approach used by Ho{\`{a}}ng et~al.~\cite{HKLSS10} to show that {\sc $k$-Colouring} is polynomial-time solvable on $P_5$-free graphs for all $k\geq 3$ (extending the analogous result of Randerath et~al.~\cite{RST02} for {\sc $3$-Colouring}).
Let us first outline the proof of~\cite{HKLSS10}.

The approach of Ho{\`{a}}ng et~al.~\cite{HKLSS10} to solve {\sc $k$-Colouring} for $P_5$-free graphs for any integer~$k$ uses Lemma~\ref{lem:P5free-clique-P3-dom} as a starting point, just as the approach of Randerath et~al.~\cite{RST02} does for the $k=3$ case.
Lemma~\ref{lem:P5free-clique-P3-dom} implies that every $k$-colourable $P_5$-free graph~$G$ has a dominating set~$D$ of size at most~$k$ (as the clique number of a $k$-colourable graph is at most~$k$).
Fix an ordering $D=\{v_1,\ldots,v_{|D|}\}$.
Then decompose the set of vertices not in~$D$ into~$|D|$ ``layers'' so that the vertices in a layer~$i$ are adjacent to~$v_i$ (and possibly to~$v_j$ for $j>i$) but not to any~$v_h$ with $h<i$.
Using the $P_5$-freeness of~$G$ to analyse the adjacencies between different layers, it is possible to branch in such a way that a polynomial number of instances of {\sc $(k-1)$-Colouring} are obtained.
Hence, by repetition, a polynomial number of instances of {\sc $2$-Colouring} are reached, which can each be solved in polynomial time due to the result of~\cite{Ed86}.

The algorithm of~\cite{HKLSS10} works by considering the more general {\sc List $k$-Colouring} problem, where each vertex~$v$ is assigned a list $L(v) \subseteq \{1,\ldots,k\}$ of permitted colours and the question is whether there is a colouring in which each vertex is assigned a colour from its list.
The algorithm immediately removes any vertices whose lists have size~$1$ at any point (and then adjusts the lists of admissible colours of all neighbours of such vertices).
We will follow the approach of~\cite{HKLSS10}.
In our case $k=3$, but we cannot remove any vertices whose lists contain a singleton colour if this colour is~$2$ or~$3$.
To overcome this extra complication we carefully analyse the $4$-vertex cycles in the graph after observing that these cycles are the only obstacles that may prevent a $3$-colouring of a $P_5$-free graph from being semi-acyclic.

For a subset $S\subseteq V(G)$ of the vertex set of a graph~$G$, we let~$G[S]$ denote the subgraph of~$G$ induced by~$S$.

\begin{theorem}\label{t-P5free-list-semi-acyclic}
{\sc List Semi-Acyclic $3$-Colouring} is solvable on $P_5$-free graphs in $O(n^{16})$ time.
\end{theorem}

\begin{proof}
\setcounter{ctrclaim}{0}
Consider an input $(G,L)$ for the problem such that~$G$ is $P_5$-free.
Since the problem can be solved component-wise, we may assume that~$G$ is connected.
If~$G$ contains a~$K_4$, then it is not $3$-colourable and the input is a \texttt{no}-instance.
As we can test whether or not~$G$ contains a~$K_4$ in $O(n^4)$ time, we now assume that~$G$ is $K_4$-free.
We may also assume that~$G$ contains at least three vertices, otherwise the problem can be trivially solved.

For $i \in \{1,2,3\}$ let $G_i=G[\{v \in V(G) \; | \; i \notin L(v)\}]$.
We apply the following propagation rules exhaustively, and, later in the proof, every time we branch on possibilities, we assume that these rules are again applied exhaustively immediately afterwards.

\begin{enumerate}[\bf Rule 1.]
\item \label{rule:nbhd-of-coloured-vertices}If $u,v \in V(G)$ are adjacent and $|L(u)|=1$, set $L(v) := L(v) \setminus L(u)$.
\item \label{rule:no-empty-lists}If $L(v)=\emptyset$ for some $v \in V(G)$, return \texttt{no}.
\item \label{rule:two-colours-bipartite}If~$G_i$ is not bipartite for some $i \in \{1,2,3\}$, return \texttt{no}.
\item \label{rule:no-2-3-C_4s}If~$G_1$ contains an induced~$C_4$, return \texttt{no}.
\item \label{rule:sort-easy-C_4s}If~$G$ contains an induced~$C_4$, and exactly one vertex~$v$ of this cycle has a list containing the colour~$1$, set $L(v)=\{1\}$.
\end{enumerate}

\noindent
We must show that these rules are safe.
That is, that when they modify the instance they do not affect whether or not it is a \texttt{yes}-instance or a \texttt{no}-instance, and when they return the answer \texttt{no}, this is correct and no semi-acyclic colouring that respects the lists can exist.
This is trivial for Rules~\ref{rule:nbhd-of-coloured-vertices} and~\ref{rule:no-empty-lists}.
We may apply Rule~\ref{rule:two-colours-bipartite} since in any $3$-colouring of~$G$ every pair of colour classes must induce a bipartite graph.
We may apply Rules~\ref{rule:no-2-3-C_4s} and~\ref{rule:sort-easy-C_4s} since in every solution, every induced~$C_4$ must contain at least one vertex coloured with colour~$1$.
In fact, if there is a $3$-colouring of~$G$ with a cycle made of vertices coloured only~$2$ and~$3$, then this cycle must be an even cycle.
Since~$G$ is $P_5$-free, such a cycle must in fact be isomorphic to~$C_4$.
Hence the problem, when restricted to $P_5$-free graphs, is equivalent to testing whether~$G$ has a $3$-colouring respecting the lists such that every induced~$C_4$ contains at least one vertex coloured with colour~$1$.

By Lemma~\ref{lem:P5free-clique-P3-dom}, $G$ has a dominating set~$S$ that is either a clique or induces a~$P_3$.
If~$S$ is a clique, then it has at most three vertices, as~$G$ is $K_4$-free, so we can find such a set in~$O(n^4)$ time.
Thus, adding vertices arbitrarily if necessary, we may assume $S=\{a_1,a_2,a_3\}$.
We consider all possible combinations of colours that can be assigned to the vertices in~$S$, that is, we branch into at most~$3^3$ cases, in which $a_1$, $a_2$ and~$a_3$ have each received a colour, or equivalently, have had their list of permissible colours reduced to size exactly~$1$.
In each case we proceed as follows.

Assume that $L(a_1)=\{c_1\}$, $L(a_2) = \{c_2\}$ and $L(a_3)=\{c_3\}$ and again apply the propagation rules above.
Partition the vertices of $V \setminus S$ into three parts $V_1$, $V_2$,~$V_3$: let~$V_1$ be the set of neighbours of~$a_1$ in $V \setminus S$, let~$V_2$ be the set of neighbours of~$a_2$ in $V \setminus S$ that are not adjacent to~$a_1$, and let $V_3=V(G) \setminus (S \cup V_1 \cup V_2)$ (see also \figurename~\ref{fig:Vi}).
Each vertex in~$V_3$ is non-adjacent to~$a_1$ and~$a_2$, so it is adjacent to~$a_3$, as~$S$ is dominating.
For $i\in \{1,2,3\}$, if $v \in V_i$, then $L(v) \subseteq \{1,2,3\} \setminus \{c_i\}$ by Rule~\ref{rule:nbhd-of-coloured-vertices}, so each vertex has at most two colours in its list.
For $i \in \{1,2,3\}$ let~$V_i'$ be the subset of vertices~$v$ in~$V_i$ with $L(v) = \{1,2,3\} \setminus \{c_i\}$.
Recall that for $i \in \{1,2,3\}$, we defined $G_i=G[\{v \in V(G) \; | \; i \notin L(v)\}]$.
Since for every $i \in \{1,2,3\}$, every vertex of~$V_i$ belongs to~$G_{c_i}$, it follows that $V_1$, $V_2$ and~$V_3$ each induce a bipartite graph in~$G$ by Rule~\ref{rule:two-colours-bipartite}.
Therefore, we may partition each~$V_i'$ into two (possibly empty) independent sets~$V_i''$ and~$V_i'''$.

\begin{figure}[h]
\begin{center}
\tikzstyle{vertex}=[circle,draw=black, fill=black, minimum size=5pt, inner sep=1pt]
\tikzstyle{edge} =[draw,black]
\begin{tikzpicture}

\foreach \pos/\name/\number/\posl in {
(3,0)/V/3/below,
(3,2)/V/2/above,
(3,4)/V/1/below}
{
\begin{scope}[shift={\pos}]
\draw[rounded corners=10pt, thick, black!50!white] (0.5,-0.5) rectangle (3.5,0.5);
\node[vertex] (\name\number1) at (1,0) {};
\node[vertex] (\name\number2) at (2,0) {};
\node[vertex] (\name\number3) at (3,0) {};
\node at (3.7,0) {$\name_\number$};
\end{scope}
}

\draw[rounded corners=10pt, thick, black!50!white] (0.5,-1.5) rectangle (1.5,3.5);
\node at (1,3.7) {$S$};

\node[vertex] (v3) at (1,-1) {};
\node [left, align=center] at (v3) {$a_3$};
\node[vertex] (v2) at (1,1) {};
\node [left, align=center] at (v2) {$a_2$};
\node[vertex] (v1) at (1,3) {};
\node [left, align=center] at (v1) {$a_1$};

\foreach \numa in {1,2,3} {
\foreach \numb in {1,2,3} {
\path[edge, black, very thick] (v\numa) -- (V\numa\numb);
}
}

\foreach \numa/\numb in {1/2,1/3,2/3} {
\foreach \numc in {1,2,3} {
\path[edge, black, very thick, dashed] (v\numa) -- (V\numb\numc);
}
}
\end{tikzpicture}
\end{center}
\caption{The sets $S$, $V_1$, $V_2$ and~$V_3$.
Dashed lines indicate edges that are not present.
Edges that are not shown may or may not be present.
In particular, vertices in~$V_1$ can be adjacent to~$v_2$ or~$v_3$ and vertices in~$V_2$ can be adjacent to~$v_3$.}\label{fig:Vi}
\end{figure}
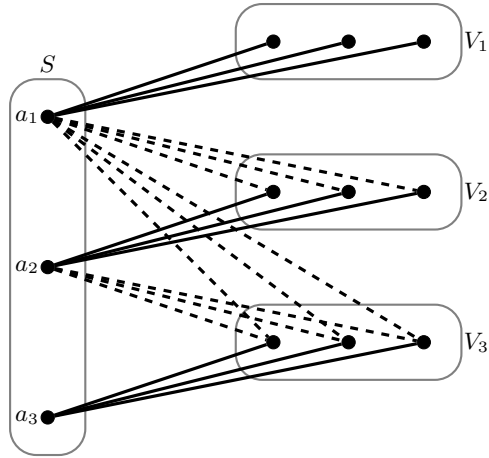

\begin{sloppypar}
Our strategy is to reduce the instance $(G,L)$ to a polynomial number of instances~$(G,L')$, in which there are no edges between any two distinct sets~$V_i'$ and~$V_j'$ (defined with respect to~$L'$).
We will do this by branching on possible partial colourings in such a way that afterwards there are no edges between~$V_i''$ and~$V_j'''$, no edges between~$V_i''$ and~$V_j''$ and no edges between~$V_i'''$ and~$V_j'''$ for every pair $i,j\in \{1,2,3\}$ with $i\neq j$.
As the branching procedure is similar for each of these possible combinations, we pick an arbitrary pair, namely~$V_1''$ and~$V_2''$.
As we shall see, we do not remove any edges between~$V_1''$ and~$V_2''$.
Instead, we decrease the lists of some of their vertices to size~$1$, so that these vertices will leave $V_1'\cup V_2'$ by definition of~$V_1'$ and~$V_2'$ (and therefore leave~$V_1''$ and~$V_2''$ by definition of~$V_1''$ and~$V_2''$).
\end{sloppypar}

Suppose that $G[V_1''\cup V_2'']$ contains an induced~$2P_2$ (see \figurename~\ref{fig:2P2}) with edges~$uu'$ and~$vv'$ for $u,v\in V_1''$ and $u',v'\in V_2''$.
Then $G[\{u',u,a_1,v,v'\}]$ is a~$P_5$, a contradiction.
It follows that $G[V_1''\cup V_2'']$ is a $2P_2$-free bipartite graph, that is, the edges between~$V_1''$ and~$V_2''$ form a chain graph, which means that the vertices of~$V_1''$ can be linearly ordered by inclusion of neighbourhood in~$V_2''$.
In other words, we fix an ordering $V_1''=\{u_1,\ldots,u_k\}$ such that $N_{V_2''}(u_1) \supseteq \cdots \supseteq N_{V_2''}(u_k)$.

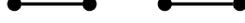
\begin{figure}
\begin{center}
\tikzstyle{vertex}=[circle,draw=black, fill=black, minimum size=5pt, inner sep=1pt]
\tikzstyle{edge} =[draw,black]
\begin{tikzpicture}
\foreach \num in {0,1,2,3} \node[vertex] (x0\num) at (\num,0) {};
\foreach \source/\dest in {x00/x01, x02/x03} \path[edge, black, very thick] (\source) -- (\dest);
\end{tikzpicture}
\end{center}
\caption{\label{fig:2P2}The graph~$2P_2$.}
\end{figure}

We choose an arbitrary colour $c' \in \{1,2,3\} \setminus \{c_1,c_2\}$.
Note that if $c_1 \neq c_2$ then this choice is unique and otherwise there are two choices (as we will show, it suffices to branch on only one choice).
Also note that every vertex in~$V_1''$ and~$V_2''$ has colour~$c'$ in its list.

We now branch over $k+1$ possibilities, namely the possibilities that vertex~$u_i$ is the first vertex coloured with colour~$c'$ (so vertices $u_1,\ldots,u_{i-1}$, if they exist, do not get colour~$c'$) and the remaining possibility that no vertex of~$V_1''$ is coloured with colour~$c'$.
To be more precise, for branch $i=1$ we set $L(u_1)=\{c'\}$, for each branch $2 \leq i \leq k$ we remove colour~$c'$ from each of $L(u_1),\ldots,L(u_{i-1})$ and set $L(u_i)=\{c'\}$ and for branch $i=k+1$ we remove colour~$c'$ from each of $L(u_1),\ldots,L(u_{k})$.
If $i=k+1$, all vertices of~$V_1''$ will have a unique colour in their list and thus leave~$V_1'$ and thus~$V_1''$ by definition of~$V_1'$.
Hence, $V_1''$ becomes empty and thus, as required, we no longer have edges between~$V_1''$ and~$V_2''$.
Otherwise, if $i \leq k$, then all of $u_1,\ldots,u_i$ will have a list containing exactly one colour, so they will leave~$V_1'$ and therefore~$V_1''$.
By Rule~\ref{rule:nbhd-of-coloured-vertices} all neighbours of~$u_i$ in~$V_2''$ will have~$c'$ removed from their lists, so they will leave~$V_2'$ and therefore~$V_2''$.
By the ordering of neighbourhoods of vertices in~$V_1''$, this means that no vertex remaining in~$V_1''$ has a neighbour remaining in~$V_2''$, so if $i\leq k$, then it is also the case that we no longer have edges between~$V_1''$ and~$V_2''$.

\medskip
\noindent
Note that removing all the edges between distinct sets~$V_i'$ and~$V_j'$ in the above way involves branching into~$O(n^{12})$ cases.
We consider each case separately, and for each case we proceed as below.

By the above branching we may assume that there are no edges between any two distinct sets~$V_i'$ and~$V_j'$.
We say that an induced~$C_4$ is {\em tricky} if there exists a (proper) colouring of it (not necessarily extendable to all of~$G$) using only the colours~$2$ and~$3$ such that every vertex receives a colour from its list.
We say that a vertex in an induced~$C_4$ is {\em good} for this induced~$C_4$ if its list contains the colour~$1$.
By definition of tricky, every good vertex for a tricky induced~$C_4$ must belong to $V_1' \cup V_2' \cup V_3'$.
By Rules~\ref{rule:no-2-3-C_4s} and~\ref{rule:sort-easy-C_4s}, every tricky induced~$C_4$ must contain at least two good vertices.
If an induced~$C_4$ contains two good vertices that are adjacent, then they must belong to the same set~$V_i'$ (since there are no edges between any two distinct sets~$V_i'$ and~$V_j'$), so they must have the same list.
This means that in every colouring of this induced~$C_4$ that respects the lists, one of the good vertices in this induced~$C_4$ will be coloured with colour~$1$, contradicting the definition of tricky.
We conclude that every tricky induced~$C_4$ must contain exactly two good vertices, which must be non-adjacent.

Suppose~$G$ contains a tricky induced~$C_4$ on vertices $v_1,v_2,v_3,v_4$, in that order, such that~$v_1$ and~$v_3$ are good.
Since the~$C_4$ is tricky, we must either have:

\begin{itemize}
\item $2 \in L(v_1)$, $3 \in L(v_2)$, $2 \in L(v_3)$ and $3 \in L(v_4)$ or
\item $3 \in L(v_1)$, $2 \in L(v_2)$, $3 \in L(v_3)$ and $2 \in L(v_4)$.
\end{itemize}

\noindent
Since~$v_2$ and~$v_4$ are not good, and there are no edges between distinct sets of the form~$V_i'$, the above implies that one of the following must hold:

\begin{itemize}
\item $L(v_1)=\{1,2\}$, $L(v_2)=\{3\}$, $L(v_3)=\{1,2\}$ and $L(v_4)=\{3\}$ or
\item $L(v_1)=\{1,3\}$, $L(v_2)=\{2\}$, $L(v_3)=\{1,3\}$ and $L(v_4)=\{2\}$.
\end{itemize}

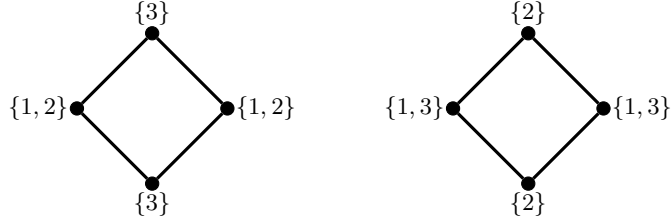
\begin{figure}
\begin{center}
\scalebox{1.0}{
\tikzstyle{vertex}=[circle,draw=black, fill=black, minimum size=5pt, inner sep=1pt]
\tikzstyle{edge} =[draw,black]
\begin{tikzpicture}
   \foreach \pos/\name / \label / \posn / \dist in {
{(0:1.0)/c11/$\{1,2\}$/{right}/0},
{(90:1.0)/c12/$\{3\}$/{above}/0},
{(180:1.0)/c13/$\{1,2\}$/{left}/0},
{(270:1.0)/c14/$\{3\}$/{below}/0}}
       { \node[vertex] (\name) at \pos {};
       \node [\posn=\dist, align=center] at (\name) {\label};
       }

\foreach \source/ \dest  in {c11/c12, c12/c13, c13/c14, c14/c11}
       \path[edge, black,  very thick] (\source) --  (\dest);

\begin{scope}[xshift=5cm]
   \foreach \pos/\name / \label / \posn / \dist in {
{(0:1.0)/c21/$\{1,3\}$/{right}/0},
{(90:1.0)/c22/$\{2\}$/{above}/0},
{(180:1.0)/c23/$\{1,3\}$/{left}/0},
{(270:1.0)/c24/$\{2\}$/{below}/0}}
       { \node[vertex] (\name) at \pos {};
       \node [\posn=\dist, align=center] at (\name) {\label};
       }

\foreach \source/ \dest  in {c21/c22, c22/c23, c23/c24, c24/c21}
       \path[edge, black,  very thick] (\source) --  (\dest);
\end{scope}

\end{tikzpicture}
}
\end{center}
\caption{Strongly tricky~$C_4$s.}\label{fig:strongly-tricky}
\end{figure}

\noindent
We say that an induced~$C_4$ is {\em strongly tricky} if its vertices have lists of this form (see also \figurename~\ref{fig:strongly-tricky}).
Note that, by the above arguments, we may assume that all tricky induced~$C_4$s in the instances we consider are in fact strongly tricky.

For $S \subsetneq \{1,2,3\}$, let~$L_S$ denote the set of vertices~$v$ with $L(v)=S$ (to simplify notation, we will write~$L_i$ instead of~$L_{\{i\}}$ and~$L_{i,j}$ instead of~$L_{\{i,j\}}$ wherever possible).
Note that for distinct sets $S,T \subseteq \{1,2,3\}$ with $|S|=|T|=2$, no vertex in~$L_S$ can have a neighbour in~$L_T$, because such vertices would be in different sets~$V_i'$, and therefore cannot be adjacent by our branching.
By Rule~\ref{rule:nbhd-of-coloured-vertices}, if $S \subsetneq T \subsetneq \{1,2,3\}$ with $|S|=1$ and~$|T|=2$, then no vertex in~$L_S$ can have a neighbour in~$L_T$.
From the above two arguments it follows that if a vertex is in $L_{1,2}$, $L_{2,3}$ or~$L_{1,3}$, then all its neighbours outside this set must be in $L_3$, $L_1$ or~$L_2$, respectively (see also \figurename~\ref{fig:Lsets}).

\begin{figure}
\begin{center}
\scalebox{1.0}{
\tikzstyle{vertex}=[circle,draw=black, fill=white, minimum size=2.4em, inner sep=1pt]
\tikzstyle{edge} =[draw,black]
\begin{tikzpicture}
%2.73205080757=sqrt(3)+1
\foreach \pos/\name / \label / \dist in {
{(90+0:1.0)/l1/$L_1$/0},
{(90+120:1.0)/l2/$L_2$/0},
{(90+240:1.0)/l3/$L_3$/0},
{(90+0:2.73205080757)/l23/$L_{2,3}$/0},
{(90+120:2.73205080757)/l13/$L_{1,3}$/0},
{(90+240:2.73205080757)/l12/$L_{1,2}$/0}}
{\node[vertex] (\name) at \pos {\label};}

\foreach \source/ \dest in {l1/l2, l2/l3, l1/l3, l1/l23, l2/l13, l3/l12} \path[edge, black, very thick] (\source) -- (\dest);

\end{tikzpicture}
}
\end{center}
\caption{The possible adjacencies of vertices in the sets~$L_S$ for $S \subseteq \{1,2,3\}$.
An edge is shown between two sets if and only if it is possible for a vertex in one of the sets to be adjacent to a vertex in the other.
Every vertex in the graph has a list of size either~$1$ or~$2$.}\label{fig:Lsets}
\end{figure}
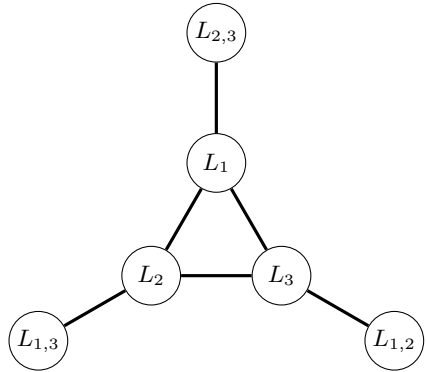

Recall that every tricky induced~$C_4$ is strongly tricky, and is therefore entirely contained in either $G[L_2 \cup L_{1,3}]$ or $G[L_3 \cup L_{1,2}]$.
By Rule~\ref{rule:two-colours-bipartite}, $G_1$ and therefore $G[L_{2,3}]$ is bipartite.
Hence we can colour the vertices of~$L_{2,3}$ with colours from their lists such that no vertex in~$L_{2,3}$ is adjacent to a vertex of the same colour in~$G$ and no induced~$C_4$s are coloured with colours alternating between~$2$ and~$3$ (indeed, recall that induced~$C_4$s cannot exist in~$G(L_{2,3})$ by Rule~\ref{rule:no-2-3-C_4s}).
It therefore remains to check whether the vertices of $G[L_2 \cup L_{1,3}]$ (and~$G[L_3 \cup L_{1,2}]$) can be coloured with colours from their lists so that no pair of adjacent vertices in~$L_{1,3}$ (resp. $L_{1,2}$) receive the same colour and every strongly tricky~$C_4$ has at least one vertex coloured~$1$.
By symmetry, it is sufficient to show how to solve the $G[L_2 \cup L_{1,3}]$ case. 
Hence we have reduced the original instance $(G,L)$ to a polynomial number of instances of a new problem, which we define below after first defining the instances.

\begin{definition}\label{d-def1}
A graph $G=(V,E)$ is {\em troublesome} if every vertex~$v$ in~$G$ has list either $L(v)=\{2\}$ or $L(v)=\{1,3\}$, such that~$L_2$ is an independent set and~$L_{1,3}$ induces a bipartite graph.
\end{definition}

\noindent
In particular, for each of our created instances the set~$L_2$ is independent due to Rule~\ref{rule:nbhd-of-coloured-vertices} and~$L_{1,3}$ induces a bipartite graph by Rule~\ref{rule:two-colours-bipartite}.
Note that by definition of troublesome, all tricky induced~$C_4$s in a troublesome graph are strongly tricky.

\begin{definition}\label{d-def2}
Let~$G$ be a troublesome graph.
A $3$-colouring of the graph~$G$ is {\em trouble-free} if each vertex receives a colour from its list, no two adjacent vertices of~$G$ are coloured alike and at least one vertex of every strongly tricky induced~$C_4$ of~$G$ receives colour~$1$.
\end{definition}

This leads to the following problem.

\problemdef{\sc Trouble-Free Colouring}{a troublesome $P_5$-free graph~$G$}{does~$G$ have a trouble-free colouring?}

\begin{sloppypar}
\noindent
We can encode an instance of {\sc Trouble-Free Colouring} as an instance of {\sc $2$-Satisfiability} as follows.
For each vertex $u \in L_{1,3}$, we create two variables~$u_{1}$ and~$u_{3}$.
If we assign~$u_1$ or~$u_3$ to be true, this means that~$u$ will be assigned colour~$1$ or~$3$, respectively.
Hence, we need to add the clauses $(u_1\vee u_3)$ and $(\overline{u_1}\vee \overline{u_3})$.
To ensure that two adjacent vertices are not coloured alike, for each pair of adjacent vertices $u,v \in L_{1,3}$ we add the clauses $(\overline{u_1} \vee \overline{v_1})$ and $(\overline{u_3} \vee \overline{v_3})$.
For each strongly tricky~$C_4$ with good vertices~$u$ and~$v$, we add the clause $(u_1 \vee v_1)$ to ensure that at least one of them will be assigned colour~$1$.
Let~${\cal I}$ be the resulting instance of {\sc $2$-Satisfiability}.
We claim that~$G$ has a trouble-free colouring if and only if~${\cal I}$ is satisfiable.
\end{sloppypar}

First suppose that~$G$ has a trouble-free colouring.
Then setting~$x_c$ to be true if the vertex~$x$ receives colour~$c$ gives a satisfying assignment for~${\cal I}$.
Now suppose that~${\cal I}$ has a satisfying assignment.
Then we colour the vertex $v \in L_{1,3}$ with colour~$c$ if~$v_c$ is set to true in this assignment.
For $v \in L_2$ we colour~$v$ with colour~$2$.
No two adjacent vertices of~$L_{1,3}$ are assigned the same colour, because~${\cal I}$ is satisfied.
No vertex of~$L_2$ is assigned the same colour as one of its neighbours, since~$L_2$ is an independent set and every vertex of~$L_{1,3}$ is assigned colour~$1$ or~$3$.
Therefore the obtained colouring is a $3$-colouring of~$G$.
Since~${\cal I}$ is satisfied, every strongly tricky~$C_4$ contains at least one vertex coloured~$1$.
Hence~$G$ has a trouble-free colouring.

\medskip
\noindent
So, by branching, we have reduced the original instance $(G,L)$ of {\sc List Semi-Acyclic $3$-Colouring} to a polynomial number of instances of {\sc $2$-Satisfiability}.
If we find that one of the instances of the latter problem is a \texttt{yes}-instance, then we obtain a corresponding \texttt{yes}-instance of {\sc Trouble-Free Colouring}.
We therefore solve {\sc Trouble-Free Colouring} on $G[L_2 \cup L_{1,3}]$ and (after swapping colours~$2$ and~$3$) on $G[L_3 \cup L_{1,2}]$.
If one of these two instances of {\sc Trouble-Free Colouring} is a \texttt{no}-instance, then we return \texttt{no} for this branch and try the next one.
If both of these are \texttt{yes}-instances, then we return \texttt{yes} and obtain a semi-acyclic $3$-colouring by combining the colourings on $G[L_1 \cup L_{2,3}]$, $G[L_2 \cup L_{1,3}]$ and (after swapping colours~$2$ and~$3$ back) $G[L_3 \cup L_{1,2}]$.
If every branch returns \texttt{no} then the original graph has no semi-acyclic $3$-colouring.
This completes the proof of the correctness of the algorithm and it remains to analyse its runtime.

Let~$n$ be the number of vertices in~$G$.
Recall that we can check if~$G$ is $K_4$-free in~$O(n^4)$ time and if it is then we can find a dominating set of size at most~$3$ in~$O(n^4)$ time.
Rule~\ref{rule:nbhd-of-coloured-vertices} can be applied in $O(n^2)$ time.
Rule~\ref{rule:no-empty-lists} can be applied in~$O(n)$ time.
Rule~\ref{rule:two-colours-bipartite} can be applied in~$O(n^2)$ time.
Rules~\ref{rule:no-2-3-C_4s} and~\ref{rule:sort-easy-C_4s} can be applied in~$O(n^4)$ time.
We first branch up to~$3^3$ times and then sub-branch~$O(n^{12})$ times and in each case we apply the rules.
It is readily seen that every created instance of {\sc $2$-Satisfiability} is solvable in~$O(n^2)$ time (see also Edwards~\cite{Ed86}).
This leads to a total runtime of $O(n^4)+O(1)\times O(n^{12}) \times (O(n^4) + O(n^2)) =O(n^{16})$.\qedllncs
\end{proof}

As mentioned, Theorem~\ref{t-P5free-list-semi-acyclic} has the following consequence.

\begin{corollary}\label{c-P5free-near-bip}
{\sc Near-Bipartiteness} can be solved in~$O(n^{16})$ time for $P_5$-free graphs.
\end{corollary}

\begin{proof}
Let~$G$ be a graph.
Set $L(v)=\{1,2,3\}$ for all $v \in V(G)$.
Then~$G$ is near-bipartite if and only if $(G,L)$ is a \texttt{yes}-instance of {\sc List Semi-acyclic $3$-Colouring}.
In particular, the vertices coloured~$1$ by a semi-acyclic colouring of~$G$ form an independent feedback vertex set of~$G$.
The corollary follows by Theorem~\ref{t-P5free-list-semi-acyclic}.\qedllncs
\end{proof}

\section{Independent Feedback Vertex Sets of $P_5$-free Graphs}\label{s-ind}

In this section we prove that {\sc Independent Feedback Vertex Set} is polynomial-time solvable for $P_5$-free graphs by extending the algorithm from Section~\ref{s-list}: the first part of our proof uses the proof of Theorem~\ref{t-P5free-list-semi-acyclic}, as we will explain in the proof of Lemma~\ref{l-part1}.
As such, we heavily use Definitions~\ref{d-def1} and~\ref{d-def2}.
Let $G=(V,E)$ be a troublesome $P_5$-free graph.
For a trouble-free colouring~$c$ of~$G$, let $t_c(G)=|\{u\in V\; |\; c(u)=1\}|$ denote the number of vertices of~$G$ coloured~$1$ by~$c$.
Let~$t(G)$ be the minimum value~$t_c(G)$ over all trouble-free colourings~$c$ of~$G$, and set~$t_c(G)=\infty$ if no such colouring exists.

\begin{lemma}\label{l-part1}
Let~$G$ be a near-bipartite $P_5$-free graph.
In~$O(n^{16})$ time it is possible to reduce the problem of finding the smallest independent feedback vertex set of~$G$ to finding the value~$t(G')$ of~$O(n^{12})$ troublesome induced subgraphs of~$G$.
\end{lemma}

\begin{proof}
Let~$G$ be a near-bipartite $P_5$-free graph, that is, we assume that~$G$ has an independent feedback vertex set.
We may assume that~$G$ is connected, otherwise, we solve the problem component-wise.
We set $L(v)=\{1,2,3\}$ for all $v \in V(G)$ and run the algorithm of Theorem~\ref{t-P5free-list-semi-acyclic}.
As can be seen from the proof of Theorem~\ref{t-P5free-list-semi-acyclic}, this algorithm branches up to~$3^3$ times and then sub-branches~$O(n^{12})$ times.
Each branch gives us, after some preprocessing in~$O(n^4)$ time, either a \texttt{no} answer, in which case we discard this branch, or two vertex-disjoint instances of {\sc Trouble-Free Colouring} (one on $G[L_2 \cup L_{1,3}]$ and one (after swapping colours~$2$ and~$3$) on $G[L_3 \cup L_{1,2}]$ and we will denote these two instances by~$G'$ and~$G''$, respectively.
Such instances consist of a troublesome graph~$G'$ or~$G''$, which is an induced subgraph of~$G$ and whose vertices have lists of admissible colours determined by the branching.

As explained in the proof of Theorem~\ref{t-P5free-list-semi-acyclic}, in any branch that we did not discard, $G[L_1 \cup L_{2,3}]$ will have a semi-acyclic $3$-colouring that respects the lists and~$L_1$ will be the set of vertices that are coloured~$1$ in any such colouring.
Therefore, given trouble-free colourings~$c$ and~$c'$ of~$G'$ and~$G''$, respectively, we can obtain an independent feedback vertex set~$S(c,c',G',G'')$ by taking the union of the set~$L_1$ of $G[L_1 \cup L_{2,3}]$ and the sets of vertices in~$G'$ and~$G''$ that~$c$ or~$c'$ colour with colour~$1$.

Now let~$c^*$ and~$c^{**}$ be such that $t(G')=t_{c^*}(G')$ and $t(G'')=t_{c^{**}}(G'')$.
If we know~$t(G')$ and~$t(G'')$, we can compute the size $s(G',G'')=|S(c^*,c^{**},G',G'')|$ in~$O(1)$ time (and the corresponding independent feedback vertex set~$S(c^*,c^{**},G',G'')$ in~$O(n)$ time).
Let~$\hat{s}$ be the minimum~$s(G',G'')$ over all branches of our procedure.
As our procedure had~$O(n^{12})$ branches, given the values of~$t(G')$ and~$t(G'')$ for every branch, we can compute~$\hat{s}$ in $O(n^{12})$~time.
As we branched in every possible way, $\hat{s}$ is the size of a minimum independent feedback vertex set of~$G$.\qedllncs
\end{proof}

We still need a polynomial-time algorithm that computes~$t(G)$ for a given troublesome $P_5$-free graph.
We present such an algorithm in the following lemma (in the proof of this lemma we again use Definitions~\ref{d-def1} and~\ref{d-def2}).

\begin{lemma}\label{l-trouble}
Let~$G$ be a troublesome $P_5$-free graph on~$n$ vertices.
Determining~$t(G)$ can be done in~$O(n^3)$ time.
\end{lemma}

\begin{proof}
\setcounter{ctrclaim}{0}
\setcounter{ctrcase}{0}
Let $G=(V,E)$ be a troublesome $P_5$-free graph.
Note that in~$G$, an induced~$C_4$ on vertices~$v_1$, $v_2$, $v_3$,~$v_4$, in that order, is strongly tricky if $v_1,v_3 \in L_{1,3}$ and $v_2,v_4 \in L_2$.

We construct an auxiliary graph~$H$ as follows.
We let $V(H)=L_{1,3}$.
Every edge of~$G[L_{1,3}]$ belongs to~$H$.
We say that such edges are {\em red}.
For non-adjacent vertices $v_1,v_3 \in L_{1,3}$, if there is a strongly tricky induced~$C_4$ on vertices~$v_1$, $v_2$, $v_3$,~$v_4$ with $v_2,v_4 \in L_2$, we add the edge~$v_1v_3$ to~$H$.
We say that such edges are {\em blue}.
Note that~$H$ is a supergraph of~$G[L_{1,3}]$ and that there exists at most one edge, which is either blue or red, between any two vertices of~$H$.

We say that a colouring of~$H$ {\em feasible} if the following two conditions are met:

\begin{enumerate}[(i)]
\renewcommand{\theenumi}{(\roman{enumi})}
\renewcommand{\labelenumi}{(\roman{enumi})}
\item\label{cond:red-not-mono}no red edge is monochromatic, that is, the two end-vertices of every red edge must be coloured, respectively, $1\&3$ or $3\&1$;\\[-11pt]
\item\label{cond:blue:not-both-3}the two end-vertices of every blue edge must be coloured, respectively, $1\&3$, $3\&1$ or $1\&1$ (the only forbidden combination is $3\&3$, as in this case we obtain a strongly tricky induced~$C_4$ in~$G$ with colours~$2$ and~$3$).
\end{enumerate}

\noindent
We note that there is a one-to-one correspondence between the set of trouble-free colourings of~$G$ and the set of feasible colourings of~$H$.
Hence, we need to find a feasible colouring of~$H$ that minimises the number of vertices coloured~$1$.

Let $R_1,\ldots, R_p$ be the components of~$G[L_{1,3}]$, or equivalently, of the graph obtained from~$H$ after removing all blue edges.
We say that these are {\em red} components.
As~$G[L_{1,3}]$ is bipartite and $P_5$-free, all red components of~$H$ are bipartite and $P_5$-free.
We denote the bipartition classes of each~$R_i$ by~$X_i$ and~$Y_i$, arbitrarily (note that these classes are unique, up to swapping their order).
We apply the following rules on~$H$ exhaustively, making sure to only apply a rule if all previous rules have been applied exhaustively.

\begin{enumerate}[\bf Rule 1.]
\item\label{r-blue1}If there is a blue edge in~$H$ between two vertices $u,v\in X_i$ or two vertices $u,v\in Y_i$, then assign colour~$1$ to~$u$ and~$v$.
\item\label{r-blue2}If there is a blue edge~$e$ in~$H$ between a vertex $u\in X_i$ and a vertex $v\in Y_i$, then delete~$e$ from~$H$.
\item\label{r-blue3}If there are blue edges~$uv$ and~$uv'$ where $u\in X_i\cup Y_i$, $v\in X_j$ and $v'\in Y_j$ ($j\neq i$), then assign colour~$1$ to~$u$.
\item\label{r-blue4}If an uncoloured vertex~$u$ is adjacent to a vertex with colour~$3$ via a blue edge, then assign colour~$1$ to~$u$.
\item\label{r-blue5}If an uncoloured vertex~$u$ is adjacent to a coloured vertex~$v$ via a red edge, then assign colour~$1$ to~$u$ if~$v$ has colour~$3$ and assign colour~$3$ to~$u$ otherwise.
\item\label{r-stop}If there is a red edge with end-vertices both coloured~$1$ or both coloured~$3$, or a blue edge with end-vertices both coloured~$3$, then return~\texttt{no}.
\item\label{r-remove-coloured}Remove all vertices that have received colour~$1$ or colour~$3$, keeping track of the number of vertices coloured~$1$.
\end{enumerate}

\noindent
Since each~$R_i$ is connected and bipartite, in every feasible colouring of~$H$, for all~$i$ either all vertices in the set~$X_i$ must be coloured~$1$ and all vertices in the set~$Y_i$ must be coloured~$3$, or vice versa.
Therefore we may safely apply Rules~\ref{r-blue1} and~\ref{r-blue2}.
Suppose that a vertex $u \in V(H)$ is incident with two blue edges~$uv$ and~$uv'$ in~$H$ for two vertices~$v$ and~$v'$ that belong to different partition classes of the same red component.
Then, as either~$v$ or~$v'$ must get colour~$3$ in every feasible colouring of~$H$, we find that~$u$ must receive colour~$1$.
Hence Rule~\ref{r-blue3} is also safe to apply.
Rules~\ref{r-blue4}--\ref{r-stop} are also safe; this follows immediately from the definition of a feasible colouring.
If a vertex~$v$ is assigned colour~$3$, then by Rule~\ref{r-blue4} all its neighbours along blue edges get colour~$1$, so Property~\ref{cond:blue:not-both-3} of a feasible colouring is satisfied for all blue edges with end-vertex~$v$.
If a vertex~$v$ is assigned a colour, then by Rule~\ref{r-blue5} all its neighbours along red edges get a different colour, so Property~\ref{cond:red-not-mono} of a feasible colouring is satisfied.
We conclude that Rule~\ref{r-remove-coloured} is safe.

By Rules~\ref{r-blue1} and~\ref{r-blue2}, if two vertices are in the same red component~$R_i$, we may assume that they are not connected by a blue edge.
Hence, we may assume from now on that red components contain no blue edges in~$H$.
By Rule~\ref{r-blue3}, we may also assume that no vertex in $V(H)\setminus V(R_j)$ is joined via blue edges to both a vertex in~$X_j$ and a vertex in~$Y_j$.

From~$H$ we construct another auxiliary graph~$H^*$ as follows.
First, we replace each red component~$R_i$ on more than two vertices by an edge~$x_iy_i$, which we say is a {\em red} edge.
Hence, the set of red components of~$H$ is reduced to a set of {\em red} components in~$H^*$ in such a way that each red component of~$H^*$ is either an edge or a single vertex.
Next, for $i\neq j$ we add an edge, which we say is a {\em blue} edge, between two vertices~$x_i$ and~$x_j$ if and only if there is a blue edge between a vertex in~$X_i$ and a vertex in~$X_j$.
Similarly, for $i\neq j$ we add a blue edge, between two vertices~$y_i$ and~$x_j$ (resp.~$y_j$) if and only if there is a blue edge between a vertex in~$Y_i$ and a vertex in~$X_j$ (resp.~$Y_j$).

Recall that, by Rules~\ref{r-blue1} and~\ref{r-blue2}, no two vertices in the same component~$R_i$ are connected by a blue edge in~$H$.
So every feasible colouring of~$H$ corresponds to a feasible colouring of~$H^*$ and vice versa.
To keep track of the number of vertices coloured~$1$, we introduce a weight function $w: V(H^*)\to \Z_+$ by setting $w(x_i)=|X_i|$ and $w(y_i)=|Y_i|$.
Our new goal is to find a feasible colouring~$c$ of~$H^*$ that minimises the sum of the weights of the vertices coloured~$1$, which we denote by~$w(c)$.

Since for each~$i$ no vertex in $V(H) \setminus V(R_i)$ is joined via blue edges to both a vertex in~$X_i$ and a vertex in~$Y_i$, we find that~$H^*$ contains no triangle consisting of one red edge and two blue edges.
As red edges induce a disjoint union of isolated edges, this means that the only triangles in~$H^*$ consist of only blue edges.
Let $B_1,\ldots, B_q$ be the components of the graph obtained from~$H^*$ after removing all red edges.
We say that these are {\em blue} components (this includes the case where they are singletons).

We will now show that all blue components of~$H^*$ are complete.

\clm{\label{clm:trouble:1}Each~$B_i$ is a complete graph.}
We prove Claim~\ref{clm:trouble:1} as follows.
For contradiction, suppose there is a blue component~$B_i$ that is not a complete graph.
Then~$B_i$ contains three vertices $u,v,w$ such that~$uv$ and~$vw$ are blue edges and~$uw$ is not a blue edge.
As~$uv$ and~$vw$ are blue edges, $v$ is not in the same red component of~$H^*$ as~$u$ or~$w$.
As no triangle in~$H^*$ can have two blue edges and one red edge, $u$ and~$w$ are not adjacent in~$H^*$, meaning that~$u$, $v$,~$w$ in fact belong to three different red components in~$H$.
Let~$u'v'$ and~$v''w'$ be blue edges of~$H$ corresponding to the edges~$uv$ and~$vw$, respectively.
As~$uw$ is not a blue edge in~$H^*$, we find that~$u'w'$ is not a blue edge in~$H'$.
We distinguish between two cases and show that neither of them is possible.

\thmcase{\label{case:trouble:1}$v'=v''$.}
As~$u'v'$ is a blue edge in~$H$, we find that in~$G$, the vertices~$u'$ and~$v'$ must have at least two common neighbours in~$L_2$.
For the same reason, in~$G$, the vertices~$v'$ and~$w'$ must have at least two common neighbours in~$L_2$.
Since~$u'w'$ is not a blue edge in~$H$, we find that in~$G$, the vertices~$u'$ and~$w'$ have at most one common neighbour in~$L_2$.
Therefore~$G$ contains two vertices $p,q \in L_2$ such that~$p$ is adjacent to~$u'$ and~$v'$ but non-adjacent to~$w'$ and~$q$ is adjacent to~$v'$ and~$w'$ but non-adjacent to~$u'$.
As~$L_2$ is an independent set in~$G$, $p$ is non-adjacent to~$q$.
Now $G[\{u', p, v', q, w'\}]$ is a~$P_5$, which is a contradiction.

\thmcase{\label{case:trouble:2}$v'\neq v''$.}
Let~$R_i$ be the red component of~$H$ containing~$v'$ and~$v''$.
Then, due to the way the red edges of~$H^*$ are constructed, either~$v'$ and~$v''$ both belong to~$X_i$, or they both belong to~$Y_i$.
As~$R_i$ is bipartite, connected and $P_5$-free, $R_i$ must contain a vertex~$s$ that is adjacent to both~$v'$ and~$v''$.
Just as in Case~\ref{case:trouble:1}, in~$G$ the vertices~$u'$ and~$v'$ have at least two common neighbours $p,p' \in L_2$, and~$v'$ and~$w'$ also have at least two common neighbours $q,q' \in L_2$.
As~$L_2$ is independent in~$G$, it follows that $\{p,p',q,q'\}$ is also an independent set (which may have size smaller than~$4$).

\begin{sloppypar}
Now~$p$ must be adjacent to at least one vertex in $\{s,v''\}$, as otherwise $G[\{u',p,v',s,v''\}]$ would be a~$P_5$.
Similarly,~$p'$ must be adjacent to at least one vertex in $\{s,v''\}$.
If~$p$ and~$p'$ are both adjacent to~$s$, then there is a blue edge between~$u'$ and~$s$ in~$H$.
This is not possible, as then we would have applied Rule~\ref{r-blue3}.
If~$p$ and~$p'$ are both adjacent to~$v''$ then there is a blue edge between~$u'$ and~$v''$ in~$H$, so Case~\ref{case:trouble:1} applies and we are done.
We may therefore assume that~$p$ is adjacent to~$s$ but non-adjacent to~$v''$, and that~$p'$ is adjacent to~$v''$ but non-adjacent to~$s$.
Similarly, we may assume that~$q$ is adjacent to~$s$ but non-adjacent to~$v'$, and that~$q'$ is adjacent to~$v'$, but non-adjacent to~$s$.
As~$p$, $p'$,~$q$ have different neighbourhoods in $\{s,v''\}$, we find that~$p$, $p'$,~$q$ are pairwise distinct.
Recalling that $\{p,p',q\} \subseteq L_2$ is an independent set, it follows that $G[\{p,v',p',v'',q\}]$ is a~$P_5$.
This contradiction completes Case~\ref{case:trouble:2}.
Hence we have proven Claim~\ref{clm:trouble:1}.
\end{sloppypar}

\medskip
\noindent
By Claim~\ref{clm:trouble:1}, $H^*$ is the disjoint union of several blue complete graphs with red edges between them.
Recall that we allow the case where these blue complete graphs contain only one vertex.
On~$H^*$ we apply the following rule exhaustively in combination with Rules~\ref{r-blue4}--\ref{r-remove-coloured}.
While doing this we keep track of the weights of the vertices coloured~$1$.

\begin{enumerate}[\bf Rule 1.]
\setcounter{enumi}{7}
\item \label{r-blue6}If there exist (red) edges~$u_1v_1$ and~$u_2v_2$ for $u_1,u_2\in B_i$ and $v_1,v_2\in B_j$ $(i \neq j)$, then assign colour~$1$ to every vertex in $(B_i\cup B_j)\setminus \{u_1,u_2,v_1,v_2\}$.
\end{enumerate}

\noindent
Since Rules~\ref{r-blue4} and~\ref{r-blue5} can be safely applied on~$H$, they can be safely applied on~$H^*$.
It follows that Rules~\ref{r-stop} and~\ref{r-remove-coloured} can also be safely applied on~$H^*$.
We may also safely apply Rule~\ref{r-blue6}: the red edges~$u_1v_1$ and~$u_2v_2$ force~$u_i$ and~$v_i$ to have different colours for $i\in\{1,2\}$, whereas the blue components forbid~$u_1,u_2$ both being coloured~$3$ and~$v_1,v_2$ both being coloured~$3$.
Hence, exactly one of~$u_1,u_2$ and exactly one of $v_1,v_2$ must be coloured~$3$.
Because at most one vertex in any blue component may be coloured~$3$, this implies that all vertices in $(B_i\cup B_j)\setminus \{u_1,u_2,v_1,v_2\}$ must be coloured~$1$.

As every vertex is incident with at most one red edge in~$H^*$, we obtain a resulting graph that is an induced subgraph of~$H^*$ with the following property: if there exist (red) edges~$u_1v_1$ and~$u_2v_2$ for $u_1,u_2\in B_i$ and $v_1,v_2\in B_j$, then $\{u_1,u_2,v_1,v_2\}$ induces a connected component of~$H^*$.
We can colour such a $4$-vertex component in exactly two ways and we remember the colouring with minimum weight (either $w(u_1)+w(v_2)$ or $w(u_2)+w(v_1)$ depending on whether~$u_1$ gets colour~$1$ or~$3$, respectively).
Hence, from now on we may assume that the resulting graph, which we again denote by~$H^*$, does not have such components.
That is, there is at most one red edge between any two blue components of~$H^*$.
As we can colour~$H^*$ component-wise, we may assume without loss of generality that~$H^*$ is connected.

For each~$B_i$ we define the subset~$B_i'$ to consist of those vertices of~$B_i$ not incident with a red edge, and we let $B_i''=B_i\setminus B_i'$.
We note the following.
If we colour every vertex of some~$B_i''$ with colour~$1$, then every neighbour of every vertex of~$B_i''$ in any other blue component~$B_j$ must be coloured~$3$ by Rule~\ref{r-blue5} (recall that vertices in different blue components are connected to each other only via red edges).
As soon as one vertex~$u$ in some blue component~$B_j$ has colour~$3$, all other vertices in $B_j-u$ must get colour~$1$ by Rule~\ref{r-blue4}.
In this way we can use Rule~\ref{r-blue4} and~\ref{r-blue5} exhaustively to {\em propagate} the colouring to other vertices of~$H^*$ where we have no choice over what colour to use.

Recall that no vertex of~$H^*$ is incident with more than one red edge.
This is a crucial fact: it implies that propagation to other blue components of~$H^*$ happens only via red edges~$vw$ between two blue components, one end-vertex of which, say~$v$, is first coloured~$1$, which implies that the other end-vertex~$w$ of such an edge must get colour~$3$; this in turn implies that all other vertices in the blue component containing~$w$ must get colour~$1$ and so on.
Hence, as~$H^*$ was assumed to be connected, colouring every vertex of a set~$B_i''$ with colour~$1$ propagates to all vertices of~$H^*$ except for the vertices of~$B_i'$.
Note that we may still colour (at most) one vertex of~$B_i'$ with colour~$3$.

Due to the above, we now do as follows for each $i\in\{1,\ldots,q\}$ in turn:
We colour every vertex of~$B_i''$ with colour~$1$ and propagate to all vertices of~$H^*$ except for the vertices of~$B_i'$.
If we obtain a monochromatic red edge or a blue edge whose end-vertices are coloured~$3$, we discard this option (by Rule~\ref{r-stop}).
Otherwise, we assign colour~$3$ to a vertex $u\in B_i'$ with maximum weight~$w(u)$ over all vertices in~$B_i'$ (if $B_i'\neq \emptyset$).
We store the resulting colouring~$c_i$ that corresponds to this option.

After doing the above for all~$q$ options, it remains to consider the cases where every~$B_i''$ contains (exactly) one vertex coloured~$3$.
Before we can use another propagation argument that tells us which vertices get colour~$3$, we first perform the following steps, only applying a step when the previous ones have been applied exhaustively.
These steps follow immediately from the assumption that every~$B_i''$ contains a vertex coloured~$3$.

\begin{enumerate}[(i)]
\renewcommand{\theenumi}{(\roman{enumi})}
\renewcommand{\labelenumi}{(\roman{enumi})}
\item \label{enum:trouble-i}Colour all vertices of every~$B_i'$ with colour~$1$ (doing this does not cause any propagation).
\item \label{enum:trouble-ii}If some~$B_i''$ consists of a single vertex, then colour this vertex with colour~$3$, and afterwards propagate by using Rule~\ref{r-blue5} exhaustively.
\item \label{enum:trouble-iii}Remove coloured vertices using Rule~\ref{r-remove-coloured}.
\end{enumerate}

\noindent
If due to~\ref{enum:trouble-ii} we obtain a monochromatic red edge or a blue edge whose end-vertices are coloured~$3$, we discard this option (using Rule~\ref{r-stop}).
Otherwise, we may assume from now on that $B_i'=\emptyset$, so $B_i''=B_i$ due to~\ref{enum:trouble-i} and that $|B_i|\geq 2$ due to~\ref{enum:trouble-ii}.
Note that doing~\ref{enum:trouble-iii} does not disconnect the graph: the vertices in the vertices in~$B_i'$ that are coloured in~\ref{enum:trouble-i} only have neighbours in the clique~$B_i$ (and these are via blue edges) and if a vertex of~$v \in B_i''$ is coloured with colour~$3$ in~\ref{enum:trouble-ii}, then its only neighbour~$w$ (via a red edge) is in a set~$B_j''$ and since~\ref{enum:trouble-i} has been applied exhaustively, the only other neighbours of~$w$ are in~$B_j''$ (via blue edges), so the propagation stops there and the graph does not become disconnected.

By our procedure, every vertex of every blue component~$B_i$ is incident with a red edge, so the total number of outgoing red edges for each~$B_i$ is equal to $|B_i|\geq 2$, and all outgoing red edges go to~$|B_i|$ different blue components.
Hence the graph~$H'$ obtained from~$H^*$ by contracting each blue component to a single vertex has minimum degree at least~$2$.
As~$H'$ has minimum degree at least~$2$, we find that~$H'$ contains an edge that is not a bridge (a bridge in a connected graph is an edge whose removal disconnects the graph).
Let~$uv$ be the corresponding red edge in~$H^*$, say~$u$ belongs to~$B_i$ and~$v$ belongs to~$B_j$.

We have two options to colour~$u$ and~$v$, namely by $1,3$ or $3,1$.
We try them both.
Suppose we first give colour~$1$ to~$u$.
Then we propagate in the same way as before.
Because~$uv$ is not a bridge in~$H'$, eventually we propagate back to~$B_i$ by giving colour~$3$ to an uncoloured vertex of~$B_i$.
When that happens we have ``identified'' the colour-$3$ vertex of~$B_i$ and then need to colour all other vertices of~$B_i$ with colour~$1$.
This means that we can in fact propagate to all blue components of~$H^*$, just as before.
If at some point we obtain a monochromatic red edge or a blue edge with end-vertices coloured~$3$, then we discard this option (by Rule~\ref{r-stop}).
Next, we give colour~$1$ to~$v$ and proceed similarly.

At the end we have at most $q+2$ different feasible colourings of~$H^*$.
We pick the one with minimum weight and translate the colouring to a feasible colouring of~$H$.
Finally, we translate the feasible colouring of~$H$ to a trouble-free colouring of the original graph~$G$.

It remains to analyse the runtime.
Let~$n$ be the number of vertices in~$G$.
Given two non-adjacent vertices in~$L_{1,3}$, we can test whether they have have two common neighbours in~$L_2$ in~$O(n)$ time.
Therefore we can construct~$H$ in~$O(n^3)$ time.

Applying Rules~\ref{r-blue1} and~\ref{r-blue2} takes~$O(n^2)$ time.
Applying Rule~\ref{r-blue3} takes~$O(n^3)$ time.
Rules~\ref{r-blue1}--\ref{r-blue3} only need to be applied exhaustively once, just after~$H$ is first constructed.
Rules~\ref{r-blue4} and~\ref{r-blue5} can be applied exhaustively in~$O(n^3)$ time.
Rule~\ref{r-stop} can be applied in~$O(n^2)$ time.
Rule~\ref{r-remove-coloured} can be applied in~$O(n)$ time.

Constructing~$H^*$ takes~$O(n^2)$ time.
By Claim~\ref{clm:trouble:1}, in~$H^*$ every blue component is a clique, so Rule~\ref{r-blue4} can be applied exhaustively on~$H^*$ in~$O(n^2)$ time.
By construction, every red component of~$H^*$ contains at most one edge, so applying Rules~\ref{r-blue5} and~\ref{r-blue6} on~$H^*$ can be done in~$O(n^2)$ time.
Therefore, Rules~\ref{r-blue4}--\ref{r-blue6} can be applied to~$H^*$ in~$O(n^2)$ time.
It follows that each option of colouring the vertices of some~$B_i''$ with colour~$1$ and then doing the propagation and colouring the vertices of~$B_i'$ takes~$O(n^2)$ time.
Since there are $q\leq n$ blue components, the total time for this is~$O(n^3)$.
Then afterwards we consider the situation where each blue component of~$H^*$ has exactly one vertex coloured~$3$.

We construct~$H'$ in~$O(n^2)$ time and also identify a non-bridge of~$H'$ in~$O(n^2)$ time.
Colouring the corresponding red edges in both ways and doing the propagation takes~$O(n^2)$ time again.
Then, if there is at least one possibility for which we did not return a \texttt{no}-answer, then we have obtained~$O(n)$ different feasible colourings of~$H^*$.
Finding the colouring with minimum weight and translating this colouring into a feasible colouring of~$H$ and then into a trouble-free colouring of the original graph~$G$ also takes~$O(n^2)$ time.\qedllncs
\end{proof}

We are now ready to state and prove the main result of our paper.

\begin{theorem}\label{t-reallymain}
The size of a minimum independent feedback vertex set of a $P_5$-free graph on~$n$ vertices can be computed in~$O(n^{16})$ time.
\end{theorem}

\begin{proof}
Let~$G$ be a $P_5$-free graph on~$n$ vertices.
As we can check in~$O(n^{16})$ time whether or not~$G$ is near-bipartite, we may assume without loss of generality that~$G$ is near-bipartite.
By Lemma~\ref{l-part1}, in~$O(n^{16})$ time we can reduce this problem to solving~$O(n^{12})$ instances of {\sc Trouble-Free Colouring} on induced subgraphs of~$G$.
By Lemma~\ref{l-trouble}, we can solve each of these instances in~$O(n^3)$ time.
The result follows.\qedllncs
\end{proof}

\noindent
{\bf Remark 1.} From our proof, we can find in polynomial time not just the size of a minimum independent feedback vertex set, but also the set itself.
The corresponding algorithm can also be adapted to find in polynomial time a maximum independent feedback vertex of a $P_5$-free graph,
or an independent feedback vertex set of arbitrary fixed size (if one exists).

\section{Independent Odd Cycle Transversal}\label{a-d}

Recall that an (independent) set $S\subseteq V$ of a graph~$G$ is an (independent) odd cycle transversal if $G-S$ is bipartite.
We also recall that a graph~$G$ has an independent odd cycle transversal if and only if~$G$ is $3$-colourable.
This means that if {\sc $3$-Colouring} is \NP-complete for a graph class~${\cal G}$, then so is {\sc Independent Odd Cycle Transversal}.
Hence, as {\sc $3$-Colouring} is \NP-complete for graphs of girth at least~$g$ for any constant~$g\geq 3$~\cite{EHK98} (see also~\cite{KKTW01,KL07}) and for line graphs~\cite{Ho81}, we find the following result.

\begin{proposition}\label{p-odd}
{\sc Independent Odd Cycle Transversal} is \NP-complete for
\begin{itemize}
\item graphs of girth at least~$g$ for any constant~$g\geq 3$;
\item for line graphs.
\end{itemize}
\end{proposition}

As shown by Chiarelli et~al.~\cite{CHJMP17}, {\sc Odd Cycle Transversal} is also \NP-complete for graphs of girth at least~$g$ for any constant~$g\geq 3$ and for line graphs.
Hence, both problems are \NP-complete for $H$-free graphs if~$H$ contains a cycle or a claw.

Our algorithm for {\sc Independent Feedback Vertex Set} restricted to $P_5$-free graphs can also be used to show that {\sc Independent Odd Cycle Transversal} is polynomial-time solvable for $P_5$-free graphs.
We just have to replace those steps from the algorithm that check whether the vertices minus the independent set (that is, the vertices coloured~$2$ and~$3$) induce a forest by steps that check whether these vertices form a bipartite graph.

\begin{theorem}\label{t-odd}
The size of a minimum independent odd cycle transversal of a $P_5$-free graph on~$n$ vertices can be computed in~$O(n^{16})$ time.
\end{theorem}

\section{Conclusions}\label{s-con}

Our main result is that {\sc Independent Feedback Vertex Set} is polynomial-time solvable for $P_5$-free graphs.
As explained in Section~\ref{a-d}, our algorithm can be readily adapted to also solve {\sc Independent Odd Cycle Transversal} for $P_5$-free graphs in polynomial time.
We also proved that {\sc Independent Feedback Vertex Set} is \NP-complete for $H$-free graphs if~$H$ contains a cycle or a claw.
As discussed, the same hardness results were known for {\sc Feedback Vertex Set} and {\sc $3$-Colouring}, and the hardness results for {\sc $3$-Colouring} immediately transfer across to {\sc Independent Odd Cycle Transversal}.

Another problem that is closely related to {\sc $3$-Colouring} is {\sc Independent Vertex Cover}, which is the independent problem variant of {\sc Vertex Cover}.
The latter problem is that of testing whether or not a given graph~$G$ has a set~$S$ of size at most~$k$ for some given integer~$k$, such that the vertices of $G-S$ form an independent set.
Similarly, the {\sc Independent Vertex Cover} problem requires~$S$ to be an independent set and is equivalent to asking whether or not a graph has a $2$-colouring such that one colour class has size at most~$k$.
This problem is clearly solvable in polynomial time.
In contrast, {\sc Vertex Cover} is \NP-complete for graphs of girth at least~$g$ for any constant $g\geq 3$~\cite{Po74}, but {\sc Vertex Cover} stays polynomial-time solvable for claw-free graphs~\cite{Mi80,Sh80}.

Apart from {\sc Independent Vertex Cover}, the complexities of the other problems that we discussed are not settled for $H$-free graphs when~$H$ is a linear forest (disjoint union of one or more paths), or even when~$H$ is a path.
Randerath and Schiermeyer~\cite{RS04} proved that {\sc $3$-Colouring} is polynomial-time solvable for $P_r$-free graphs for $r=6$, and more recently, Bonomo et~al.~\cite{BCMSZ} proved this for $r=7$.
The complexity of {\sc $3$-Colouring} for $P_r$-free graphs is not known for $r\geq 8$ (we refer to~\cite{GJPS} for further details on {\sc $k$-Colouring} for $P_r$-free graphs).

\begin{sloppypar}
The problems {\sc Feedback Vertex Set} and {\sc Odd Cycle Transversal} are polynomial-time solvable for the class of permutation graphs~\cite{BK85}, which contains the class of $P_4$-free graphs~\cite{BK85}, but their complexity is not known for $P_r$-free graphs when $r\geq 5$.
This is in contrast to {\sc Independent Feedback Vertex Set} and {\sc Independent Odd Cycle Transversal} due to our result on $P_5$-free graphs.
For these two problems we do not know their complexity for $r\geq 6$.
As mentioned, Lokshantov et~al.~\cite{LVV14} proved that {\sc Vertex Cover} (or equivalently, {\sc Independent Set}) is polynomial-time solvable for $P_5$-free graphs, and this result was recently extended to $P_6$-free graphs by Grzesik~et~al.~\cite{GKPP17}.
The computational complexity of {\sc Vertex Cover} for $P_r$-free graphs is not known for $r\geq 7$.
\end{sloppypar}

We refer to Table~\ref{t-thetable} for a summary of the above problems.
In this table we also added the {\sc Dominating Induced Matching} problem, which is also known as the {\sc Efficient Edge Domination} problem.
This problem is that of deciding whether or not a graph~$G$ has an independent set~$S$ such that $G-S$ is an {\em induced matching}, that is, the disjoint union of a set of isolated edges.
Cardoso et~al.~\cite{CCDS08} proved that $G-S$ is in fact a maximum induced matching.
We note that every graph~$G$ whose vertex set allows a partition into an independent set and an induced matching is $3$-colourable.
Grinstead et~al.~\cite{GSSH93} proved that {\sc Dominating Induced Matching} is \NP-complete.
Later, the problem was shown to be \NP-complete or polynomial-time solvable for various graph classes.
In particular, Brandst{\"a}dt and Mosca~\cite{BM14} proved that {\sc Dominating Induced Matching} for $P_r$-free graphs is polynomial-time solvable if $r=7$.
Later they extended their result to $r=8$~\cite{BM17}.
The complexity status of {\sc Dominating Induced Matching} is unknown for $r\geq 9$.
Hertz et~al.~\cite{HLRZW15} conjectured that the problem is polynomial-time solvable for $H$-free graphs whenever~$H$ is a forest, each connected component is a subdivided claw, a path or an isolated vertex.
 
Completing Table~\ref{t-thetable} is a highly non-trivial task.
In particular, we note that no \NP-hardness results are known for any of the problems in Table~\ref{t-thetable} when restricted to $P_r$-free graphs.
As such, it would be interesting to know whether the problem of determining whether or not a $P_r$-free graph has an independent feedback vertex set (or equivalently, whether or not a $P_r$-free graph is near-bipartite) is polynomially equivalent to the {\sc $3$-Colouring} problem restricted to $P_{f(r)}$-free graphs for some function~$f$.

To solve {\sc Independent Feedback Vertex} on $P_r$-free graphs for $r\in \{6,7,8\}$, one could try to exploit the techniques used to solve {\sc $3$-Colouring} for $P_r$-free graphs, just as we did for the $r=5$ case in this paper.
However, this seems difficult due to additional complications and a different approach may be required.

\begin{table}[h]
\centering
\begin{tabular}{|l|l|l|l|l|l|l|}
\hline
& $r\leq 4$ & $r=5$ & $r=6$ & $r=7$ &$r=8$ &$r\geq 9$\\[-1pt]
\hline
{\sc $3$-Colouring} & P &P~\cite{RST02} &P~\cite{RS04} &P~\cite{BCMSZ} &? &?\\[-1pt]
\hline
{\sc Vertex Cover} &P &P~\cite{LVV14} &P~\cite{GKPP17} &? &? &?\\[-1pt]
\hline
{\sc Independent Vertex Cover} &P &P &P &P &P &P\\[-1pt]
\hline
{\sc Feedback Vertex Set} &P~\cite{BK85} &? &? &? &? &?\\[-1pt]	
\hline
{\sc Independent Feedback Vertex Set} &P~\cite{TIZ15} &P &? &? &? &?\\[-1pt]
\hline
{\sc Near-Bipartiteness} &P~\cite{BBKNP13} &P &? &? &? &?\\[-1pt]
\hline
{\sc Odd Cycle Transversal} &P~\cite{BK85} &? &? &? &? &?\\[-1pt]
\hline
{\sc Independent Odd Cycle Transversal} &P &P &? &? &? &?\\[-1pt]	
\hline
{\sc Dominating Induced Matching} &P &P &P &P~\cite{BM14} &P \cite{BM17} &?\\[-1pt]
\hline			
\end{tabular}
\vspace*{2mm}
\caption{The complexity of {\sc $3$-Colouring} and eight related problems for $P_r$-free graphs, where $r\geq 1$ is a fixed integer (the columns $r\leq 4$ and $r\geq 8$ represent multiple cases).
Here, P stands for being polynomial-time solvable, whereas open cases are denoted by~?.}\label{t-thetable}
\end{table}

Finally, we point out that the connected problem variants {\sc Connected Feedback Vertex Set}, {\sc Connected Odd Cycle Transversal}, and {\sc Connected Vertex Cover}, which each require the desired set~$S$ of size at most~$k$ to induce a connected graph, are also known to be \NP-complete for line graphs and graphs of arbitrarily large girth.
This was shown by Chiarelli et~al.~\cite{CHJMP17} for {\sc Connected Feedback Vertex Set} and {\sc Connected Odd Cycle Transversal}, whereas Munaro~\cite{Mu} proved that {\sc Connected Vertex Cover} is \NP-hard for line graphs (of planar cubic bipartite graphs) and for graphs of arbitrarily large girth.
Moreover, for these three problems the complexity has not yet been settled for $H$-free graphs when~$H$ is a linear forest (see~\cite{CHJMP17} for some partial results in this direction).

\bibliography{mybib}

\end{document}